%% file: tspRevision.tex
\documentclass[11pt,notitlepage]{article}

\usepackage{thumbpdf, amssymb, amsmath, amsthm, microtype, array,
graphicx, verbatim, listings, color, fancybox}
\usepackage[pdftex]{hyperref}
\usepackage{mathtools}
\usepackage{amsmath, nccmath}
\usepackage{booktabs}
\usepackage{glossaries}
\usepackage{setspace}
\usepackage{acro}
\usepackage{url}
\usepackage[linesnumbered,ruled]{algorithm2e}
\usepackage{subcaption}
\usepackage[utf8]{inputenc}
\usepackage[english]{babel}
\usepackage{caption}
\usepackage{subcaption}
\usepackage{thmtools}
\usepackage{bbm}
\usepackage{pgf,tikz}
\usepackage[margin=1.1in]{geometry}
\usepackage{multicol,multirow}
\usepackage{mathrsfs}
\usepackage{epigraph}
\usepackage{thm-restate}
\usepackage{cleveref}

\usetikzlibrary{arrows}

\newcommand\ceil[1]{\lceil#1\rceil}

{
	
\declaretheorem[name=Theorem]{thm}
\declaretheoremstyle[style=claim,qed=$\Diamond$]{claim}
\declaretheoremstyle[style=plain,qed=$\square$]{theorem}

\theoremstyle{plain}

\newtheorem{lemma}[thm]{Lemma}
\newtheorem{conjecture}{Conjecture}
\newtheorem{problem}{Open Problem}

\newtheorem{observation}[thm]{Observation}
\newtheorem{proposition}[thm]{Proposition}
\newtheorem{definition}{Definition}
\newtheorem{claim}{Claim}

\numberwithin{equation}{section}
\numberwithin{thm}{section}
\numberwithin{lemma}{section}
\numberwithin{corollary}{section}
\numberwithin{observation}{section}
\numberwithin{proposition}{section}
\numberwithin{example}{section}
\numberwithin{figure}{section}

\DeclareMathOperator{\subt}{{SEP}}

\DeclareMathOperator{\tsp}{{{TSP}}}
\DeclareMathOperator{\2ec}{{{2ECM}}}

\newcommand{\vtree}{\mathrm{\textsc{-Tree}}}
\newcommand{\join}{\mathrm{\textsc{-{Join}}}}

\newcommand{\TT}{\mathcal{T}}

\newcommand{\N}{N_2} 

\newcommand{\fu}{{f_u}}
\newcommand{\gu}{{g_u}}

\newcommand{\oU}{\overline{U}}

\newcommand{\Gux}{{G_x^{U}}}
\newcommand{\Goux}{{G_x^{\oU}}}

\newenvironment{cproof}
{\begin{proof}
 [Proof.]
 \vspace{-1.5\parsep}
}
{ \end{proof}}

\onehalfspacing

\newcommand{\arash}[1]{\noindent{\bf {\color{blue!60!black}{\sc Arash:}  #1}}}

\newcommand{\alantha}[1]{\noindent{\bf {\color{magenta!80!	black}{\sc Alantha:}  #1}}}

\begin{document}{\bibliographystyle{alpha}}


\title{Towards Improving Christofides Algorithm on Fundamental
Classes by Gluing Convex Combinations of Tours\footnote{A preliminary version of these results was
  published in the Proceedings of the 27th Annual European Symposium
  on Algorithms (ESA) 2019~\cite{HNESA}.}}

\author{\textsc{Arash Haddadan\thanks{Carnegie Mellon University,
			Pittsburgh PA 15213,
			USA. {\tt{ahaddada@andrew.cmu.edu}}.}} \and \textsc{Alantha
		Newman\thanks{CNRS and Universit\'e Grenoble Alpes, 38000, Grenoble,
			France. {\tt{alantha.newman@grenoble-inp.fr}}.}}}

\maketitle

\input{abstract}

\input{intro}

\input{notation}

\input{main}

\input{base-case}

\input{uniform-points}

\input{conclusion}

\bibliographystyle{abbrv}
\bibliography{tspjournal-ref}

\end{document}

%% file: abstract.tex
\begin{abstract}
We present a new approach for gluing tours over certain tight, 3-edge
cuts.  Gluing over 3-edge cuts has been used in algorithms for finding
Hamilton cycles in special graph classes and in proving bounds for
2-edge-connected subgraph problem, but not much was known in this
direction for gluing connected multigraphs.  We apply this approach to
the traveling salesman problem (TSP) in the case when the objective
function of the subtour elimination relaxation is minimized by a {\em
  $\theta$-cyclic point}: $x_e \in \{0,\theta, 1-\theta, 1\}$, where
the support graph is subcubic and each vertex is incident to at least
one edge with $x$-value 1.  Such points are sufficient to resolve TSP
in general.  For these points, we construct a convex combination of
tours in which we can reduce the usage of edges with $x$-value 1 from
the $\frac{3}{2}$ of Christofides algorithm to
$\frac{3}{2}-\frac{\theta}{10}$ while keeping the usage of edges with
fractional $x$-value the same as Christofides algorithm.  A direct
consequence of this result is for the Uniform Cover Problem for TSP:
In the case when the objective function of the subtour elimination
relaxation is minimized by a {\em $\frac{2}{3}$-uniform point}: $x_e
\in \{0, \frac{2}{3}\}$, we give a $\frac{17}{12}$-approximation
algorithm for TSP.  For such points, this lands us halfway between the
approximation ratios $\frac{3}{2}$ of Christofides algorithm and
$\frac{4}{3}$ implied by the famous ``four-thirds conjecture''.
\end{abstract}

%% file: intro.tex
\section{Introduction}\label{tsp}
In the \textsc{Traveling Salesperson Problem (TSP)} we are given an
integer $n\geq 3$ as the number of vertices and a non-negative cost
vector $c$ defined on the edges of the complete graph
$K_n=(V_n=\{1,\ldots,n\},E_n = {\{1,\ldots,n\}\choose 2})$. We wish to find
the minimum cost Hamilton cycle in the graph $K_n$ with respect to
costs $c$. This problem is NP-hard and it is even NP-hard to
approximate within any constant factor \cite{davids}. A natural
assumption is that the cost vector $c$ is metric: $c_{ij}+c_{jk}\geq
c_{ik}$ for $i,j,k\in V_n$.  This special case of TSP is called metric
TSP.  Metric TSP is NP-hard \cite{GareyJohnson}. In fact, metric TSP
is APX-hard and NP-hard to approximate with a ratio better than
220/219 \cite{apx-hardness}.  Since we never deal with non-metric TSP
in this paper, we henceforth refer to metric TSP by TSP.

The following linear programming relaxation for the TSP is known as
the subtour elimination relaxation.
\begin{equation*}
\min \{cx:\; \sum_{j\in V_n\setminus \{i\}} x_{ij} = 2  \text{ for } i\in V_n, \; \sum_{i\in U, j\notin U} x_{ij} \geq 2 \; \text{ for } \emptyset \subset U\subset V_n,\;  x\in [0,1]^{E_n}\}.
\end{equation*}
We let $\subt(K_n)$ denote the feasible region of this linear
programming relaxation.  Since vector $c$ is metric, any spanning, connected
Eulerian multi-subgraph of $K_n$ (henceforth a \textit{tour of $K_n$})
can be used to find a Hamilton cycle of $K_n$ of no greater
cost.\footnote{For a graph $G = (V_n, E)$, we define a \textit{tour of
$G$} to be a tour of $K_n$ that uses only edges in $E$.  Notice that
the incidence vector for such a tour lives in $\mathbb{R}^E$.} We
define $\tsp(K_n)$ to be the convex hull of incidence vectors of tours
of $K_n$.  The integrality gap of the subtour elimination relaxation
for the TSP is\footnote{ We use
$\mathbb{R}^{p}_{\geq 0}$ to denote $\{x\in\mathbb{R}^{p}, x_i \geq
0 \text{ for } i \in \{1, \dots, p\}, x\neq 0\}$.} 
\begin{equation}
g(\tsp) = \max_{n\geq 3,c\in \mathbb{R}^n_{\geq 0}} \frac{\min_{x\in \tsp(K_n)} cx}{\min_{x\in \subt(K_n)} cx}.
\end{equation}
By the characterization of the integrality gap by
Goemans~\cite{goemansblocking} (see also \cite{Carr2004}), $g(\tsp)$
can also be defined as
\begin{equation}\label{eq:IGtsp2}
g(\tsp) =\min\{\alpha: \alpha\cdot x \in \tsp(K_n) \; \text{for} \; n\geq 3, \; x\in \subt(K_n)\}.
\end{equation}

It is well-known that $g(\tsp)\geq \frac{4}{3}$.  Based on the
definition of $g(\tsp)$ in (\ref{eq:IGtsp2}), we can interpret this
lower bound as follows: for any $\epsilon>0$, there is a point $x$
such that $x\in \subt(K_n)$ and
$(\frac{4}{3}-\epsilon)x\notin \tsp(K_n)$. As for upper bounds, a
polyhedral analysis of the classical algorithm of Christofides proves
$g(\tsp)\leq \frac{3}{2}$, as well as providing a
$\frac{3}{2}$-approximation algorithm for the TSP \cite{chris,wolsey}.
 
\begin{restatable}[Polyhedral proof of Christofides
algorithm \cite{chris, wolsey}]{thm}{christofides}
		\label{3/2}
		If $x\in \subt(K_n)$, then $\frac{3}{2}x\in \tsp(K_n)$.  
\end{restatable}

For more than four decades, there was no result that showed for all
$x\in \subt(K_n)$, the vector $(\frac{3}{2}-\epsilon)x \in \tsp(K_n)$
for some constant $\epsilon >0$.  After we submitted
this journal paper, this was resolved for a tiny
value of $\epsilon$~\cite{karlin2020slightly,karlin2021slightlyArxiv}.
Motivated by the lower bound of
$\frac{4}{3}$ on $g(\tsp)$, the following has been conjectured and is
wide open.
\begin{conjecture}[The four-thirds conjecture]\label{conj:4/3} If
$x\in \subt(K_n)$, then $\frac{4}{3}x\in \tsp(K_n)$.
\end{conjecture}

Despite the lack of progress towards resolution of
Conjecture \ref{conj:4/3}, there has been great success in providing
new bounds on $g(\tsp)$ for special cases in the past
decade~\cite{oss,momkesvensson,sebovygen}.  Next we present two
equivalent formulations of Conjecture \ref{conj:4/3} that are relevant
for our results.

\subsection{Fundamental Classes for TSP}\label{intro:fundtsp}
One approach to the four-thirds conjecture is to consider fundamental
classes for TSP. Fundamental classes of points were introduced by Carr
and Ravi \cite{carrravi} and further developed by Boyd and
Carr \cite{boydcarr} and Carr and Vempala \cite{Carr2004}. A set of
vectors $\mathcal{X}$ is a \textit{fundamental class for TSP} if (i)
for every $x\in \mathcal{X}$ we have $x\in \subt(K_n)$ and (ii)
proving $\alpha\cdot x\in \tsp(K_n)$ for all $x\in \mathcal{X}$
implies $g(\tsp)\leq \alpha$.

\subsubsection{Cyclic Points}
  For $x \in \subt(K_n)$, we define $G_x = (V_n, E_x)$ to be the
subgraph of $K_n$ whose edge set corresponds to the support of $x$
(i.e., $E_x = \{e \;:\; x_e >0 \}$).\footnote{We sometimes abuse notation and treat $x$ as a vector in
$\mathbb{R}^{E_x}$.}
The set of \textit{cyclic
points} form a fundamental class for TSP with a very simple structure~\cite{Carr2004}.
\begin{definition}\label{def:cyclicpts}
	A point $x$ is called a $\theta$-cyclic point for some $0<\theta\leq\frac{1}{2}$ if:
	\begin{itemize}
		\item Vector $x$ is in $ \subt(K_n)\cap \{0,\theta,1-\theta,1\}^{E_n}$.
		\item The support graph of $x$, $G_x=(V_n,E_x)$, is subcubic.
		\item For each $v\in V_n$ there is at least one edge $e\in \delta(v)$ with $x_e = 1$.
	\end{itemize}
\end{definition}
Observe that for a $\theta$-cyclic point $x$ we have: (i) the set of
1-edges in $G_x$, $W_x=\{e: x_e = 1\}$, forms vertex-disjoint paths of
$G_x$, (ii) the fractional edges in $G_x$, $H_x=\{e:x_e<1\}$, form
vertex-disjoint
cycles of $G_x$.  It is easy to see that if $\theta < \frac{1}{2}$
all the cycles in $H_x$ have even length since $x\in \subt(K_n)$.
Conjecture \ref{conj:4/3} can be restated as follows.

\begin{conjecture}\label{half-int-tsp}
For any $\theta \in (0,1)$ and any $\theta$-cyclic point $x\in
\mathbb{R}^{E_n}$, we have $\frac{4}{3}x\in \tsp(K_n)$.
\end{conjecture}

Similarly, Theorem \ref{3/2} can also be restated as follows.
\begin{thm}
For any $\theta \in (0,1)$ and any $\theta$-cyclic point
$x\in \mathbb{R}^{E_n}$, we have $\frac{3}{2}x\in \tsp(K_n)$.  \end{thm}
The main result of this paper is to show that we can save
on the 1-edges of $\theta$-cyclic points.

\vspace*{5pt}
\noindent\fbox{%
	\parbox{\textwidth}{%
		\vspace*{1pt}
	\begin{restatable}{thm}{introSavingOneEdge2}
	\label{main}
	Let $x\in \mathbb{R}^{E_n}$ be a $\theta$-cyclic point. Define vector $y$ as
	follows: $ y_e = \frac{3}{2}-\frac{\theta}{10}$ for $e \in
	W_x$, $y_e = \frac{3}{2}x_e$ for $e\in H_x$ and $y_e = 0$ for
	$e \notin E_x$.  Then $y\in \tsp(K_n)$.
\end{restatable}
	}%
}
\vspace*{3pt}

In fact a bound on $g(\tsp)$ restricted to $\frac{1}{2}$-cyclic points
would also provide a bound on $g(\tsp)$ when restricted to all
half-integral points of the subtour elimination
relaxation~\cite{Carr2004}, a special case that has received some
attention~\cite{carrravi,boydsebo,karlin2019improved}, and was highlighted
in a recent conjecture of Schalekamp, Williamson and van Zuylen
stating that the maximum on $g(\tsp)$ is achieved for half-integral
points of the subtour elimination polytope~\cite{swz}.  For
$\frac{1}{2}$-cyclic points Theorem \ref{main} implies the following.
\begin{restatable}{corollary}{introSavingOneEdge}
	\label{maincorollary} Let $x$ be a $\frac{1}{2}$-cyclic
	point. Define vector $y$ as follows: $ y_e
	= \frac{3}{2}-\frac{1}{20}$ for $e \in W_x$, $y_e
	= \frac{3}{4}$ for $e\in H_x$ and $y_e = 0$ for $e \notin E_x$. Then $y\in \tsp(K_n)$.
\end{restatable}

Theorem \ref{main} gives an improved bound for recently
studied special case of {\em uniform points}, which form another
fundamental class for TSP.

\subsubsection{Uniform Points}\label{sec:uniformpts}

A point $x\in \subt(K_n)$ is a {\em $\frac{2}{k}$-uniform point}
for some $k\in \mathbb{Z}_{\geq 3}$ if for each edge $e
  \in E_n$, either $x_e$ is a multiple of $\frac{2}{k}$ or $x_e =
  0$. It is clear that $g(\tsp)\leq \alpha$ if and only if for all
$k\in \mathbb{Z}_{\geq 3}$ and all $\frac{2}{k}$-uniform points $x$ we
have $\alpha\cdot x \in \tsp(K_n)$. Therefore, $\frac{2}{k}$-uniform
points form a fundamental class for TSP. We can restate the
four-thirds conjecture as follows.
\begin{restatable}{conjecture}{fourthirdseq}\label{conj:4/3eq}
	For any integer $k\geq 3$ and $\frac{2}{k}$-uniform point $x\in\mathbb{R}^{E_n}$ we have $\frac{4}{3} x\in \tsp(K_n)$.
\end{restatable} 

Seb\H{o} et al. considered a weakening of Conjecture \ref{conj:4/3eq}
in the case when $k=3$~\cite{sebouniformcover}.  For a
$\frac{2}{3}$-uniform point $x\in \mathbb{R}^{E_n}$, we have
$\frac{3}{2}x\in \tsp(K_n)$ by Theorem \ref{3/2}.  Thus, they asked if
there is constant $\epsilon>0$ such that the vector
$(\frac{3}{2}-\epsilon)\cdot x\in \tsp(K_n)$.  Of course, the
four-thirds conjecture itself implies that the value of $\epsilon$ is
at least $\frac{1}{6}$.

\begin{restatable}{conjecture}{seboconj}
	\label{conj:uc-cubic}
	If $x\in \mathbb{R}^{E_n}$ is a $\frac{2}{3}$-uniform point, then $\frac{4}{3} x \in \tsp(K_n)$. 
\end{restatable}

One application of Theorem \ref{main} is to show that the value of
$\epsilon$ is at least $\frac{1}{12}$, which brings us ``halfway''
towards resolving Conjecture \ref{conj:uc-cubic}.  The proof of
Theorem \ref{uc-improved} can be found in Section
\ref{sec:thm1point5}.

\vspace*{5pt}
\noindent\fbox{%
	\parbox{\textwidth}{%
		\vspace*{2pt}
	\begin{thm}\label{uc-improved}
	Let $x\in \mathbb{R}^{E_n}$ be a $\frac{2}{3}$-uniform point. Then $\frac{17}{12} x\in \tsp(K_n)$. If $G_x$ is Hamiltonian, then $\frac{87}{68} x \in \tsp(K_n)$. 
\end{thm} 
	}%
}
\vspace*{3pt}

The following observation
was first made by Carr and Vempala \cite{Carr2004}.

\begin{proposition}\label{k-ec,k-reg}
Let $k\in \mathbb{Z}_{\geq 3}$. We have $\alpha\cdot x\in \tsp(K_n)$
for all $\frac{2}{k}$-uniform points $x$ if and only if for all
$k$-edge-connected $k$-regular multigraphs $G=(V,E)$, the point
$\alpha\cdot (\frac{2}{k}\cdot \chi^E)$ dominates a convex
combination of tours of $G$. 
\end{proposition}

Focusing on uniform points (i.e., proving Conjecture \ref{conj:4/3eq})
is a possible path towards the four-thirds conjecture.
One way to attack this conjecture is to decompose a uniform point into
a convex combination of different points with a useful structure.  For
example, 
Theorem
\ref{uc-improved} uses this approach for $\frac{2}{3}$-uniform points
by decomposing them into convex combinations of $\frac{1}{2}$-cyclic
points and applying our main result (Theorem \ref{main}).

The next step on this path is to study $\frac{2}{4}$-uniform points.
TSP on $\frac{2}{4}$-uniform points is equivalent to so-called
``half-integral'' TSP: if for any 4-edge-connected 4-regular graph
  $G$ the vector $\alpha\cdot (\frac{1}{2}\cdot \chi^{E(G)})$
  dominates a convex combination of incidence vectors of tours of $G$,
  then $g(\tsp)$ restricted to half-integral instances is at most
  $\alpha$.

For graph $G=(V,E)$, a cut $U \subset V$ is {\em{proper}} if $|U| \geq
2$ and $|V\setminus{U}| \geq 2$.  If we assume that the 4-regular
4-edge-connected graph $G$ does not contain proper 4-edge cuts (i.e.,
it is essentially 6-edge-connected), then the following theorem, whose
proof appears in Section \ref{sec:thm1point7}, is relevant.  Its proof
uses essentially the same techniques we use to prove Theorems
\ref{main} and \ref{uc-improved}, but is simpler (i.e., does not
require gluing) because there are no tight proper cuts containing
1-edges.

\begin{restatable}{thm}{introFourRegBasecase}
	\label{onlybase-intro}
	Let $G=(V,E)$ be a 4-edge-connected 4-regular graph with an even
	number of vertices and no proper 4-edge cuts. Then the vector
	$(\frac{3}{2}-\frac{1}{42})\cdot
	(\frac{1}{2}\cdot \chi^{E})$ dominates a convex combination
	of incidence vectors of tours of $G$.
\end{restatable}

Theorem \ref{onlybase-intro} could serve as the base case if we could
prove an analogous theorem when $G$ has an odd number of vertices and
if we could glue over proper 4-edge cuts of $G$.  However, the gluing
arguments we present for $\theta$-cyclic points can not easily be
extended to this case due to the increased complexity of the
distribution of patterns over 4-edge cuts.\footnote{We remark
  that
  \cite{gupta2021matroid} did extend Theorem \ref{onlybase-intro} to the
  case when $G$ has an odd number of vertices.
  This is discussed further in Section \ref{sec:related}.}

\subsection{Related Work on Fundamental Points}\label{sec:related}

Fundamental points were introduced in a series of papers for the
TSP~\cite{carrravi} and for the minimum cost 2-edge-connected multigraph
problem (2ECM)~\cite{carrravi,boydcarr,Carr2004}. Let $\2ec(K_n)$ be
the convex hull of incidence vectors of 2-edge-connected multigraphs
of $K_n$. Clearly, $\tsp(K_n)\subseteq \2ec(K_n)$.

Consider a $\frac{1}{2}$-cyclic point $x\in \mathbb{R}^{E_n}$. If
$H_x$ is a collection of 3-cycles, then Boyd and Carr showed
$\frac{4}{3}x\in \tsp(K_n)$~\cite{boydcarr} and Boyd and Legault
showed $\frac{6}{5}x \in \2ec(K_n)$~\cite{boydlegault}. Boyd and
Seb\H{o} presented a polynomial time algorithm proving that if $H_x$
is a collection of 4-cycles, then $\frac{10}{7}x\in
\tsp(K_n)$~\cite{boydsebo}. For the same class, we gave an efficient
algorithm proving $\frac{9}{7}x\in
\2ec(K_n)$~\cite{haddadan2021efficient}.  
The interest
in $\frac{1}{2}$-cyclic points stems in part from the aforementioned
conjecture that the maximum value of $g(\tsp)$ is achieved for
instances of the TSP where the optimal solution to $\min\{cx\;:\; x\in
\subt(K_n)\}$ is $\frac{1}{2}$-cyclic~\cite{swz}. In fact, in each of
the classes above there is a family of instances that achieves the
largest known lower bound on $g(\tsp)$ and the integrality gap of the
subtour elimination relaxation for 2ECM
\cite{carrravi,alexander2006integrality,boydcarr,boydsebo}.

Now we review the results on uniform points. Let $x\in
\mathbb{R}^{E_n}$ be a $\frac{2}{k}$-uniform point. Carr and Ravi
showed if $k=4$, then $\frac{4}{3}x\in
\2ec(K_n)$~\cite{carrravi}. Boyd and Legault showed that if $k=3$,
then $\frac{6}{5}x\in \2ec(K_n)$~\cite{boydlegault}.  Legault later
improved the factor $\frac{6}{5}$ to
$\frac{7}{6}$~\cite{philipe-masters}. Haddadan, Newman and Ravi proved
that for $k=3$, we have $\frac{27}{19}x\in
\tsp(K_n)$~\cite{uniform}. Boyd and Seb\H{o} showed that if $G_x$ is
additionally Hamiltonian, then $\frac{9}{7}x\in
\tsp(K_n)$~\cite{boydsebo}.

\paragraph{Recent Results:} Since the conference version~\cite{HNESA} of this paper was announced,
there have been many developments.
First, Karlin, Klein and
Oveis Gharan showed that for any $\frac{1}{2}$-cyclic point $x\in
\mathbb{R}^{E_n}$, we have $(\frac{3}{2}-\epsilon)x\in \tsp(K_n)$ for
some constant $\epsilon >0$~\cite{karlin2019improved}.  Subsequently,
Gupta et al.~\cite{gupta2021matroid} combined our
  approach with that of \cite{karlin2019improved} resulting in a
  further improved approximation factor for this case.  Essentially,
  they extend Theorem \ref{onlybase-intro} to the
  odd case and are then able to substitute this combinatorial
  construction for the max-entropy approach used
  by \cite{karlin2019improved} in the base case.  They
  then use other methods (from \cite{karlin2019improved}) in lieu of
  gluing.  With respect to uniform points, if $x \in \mathbb{R}^{E_n}$
  is a $\frac{2}{k}$-uniform point, then 
these results also imply
  that when $k=4$, 
$(\frac{3}{2}-\epsilon)x\in \tsp(K_n)$ for
some constant $\epsilon>0$~\cite{karlin2019improved,gupta2021matroid}.  
Finally, another recent result showed that $(\frac{3}{2}-\epsilon)x\in
\tsp(K_n)$ for all $k$ for some tiny value of
$\epsilon$~\cite{karlin2020slightly,karlin2021slightlyArxiv}.

\subsection{Gluing Convex Combinations Over 3-edge Cuts}

A key part of our proof of Theorem \ref{main} is gluing solutions over
certain 3-edge cuts, thereby reducing to instances without such cuts,
which are easier to solve.  This approach of gluing solutions over
3-edge cuts, and thereby reducing to a problem on graphs without
proper 3-edge cuts was first introduced by Cornu\'{e}jols, Naddef and
Pulleyblank \cite{cornuejols1985traveling}.  For a graph $G=(V,E)$,
let $U \subset V$ and denote by $G^U$ the graph obtained by
contracting $U$ (i.e., identifying all vertices in $U$ to a single
vertex and removing the resulting loops).  Cornu\'ejols et al. defined
a class of 3-edge-connected graphs $\mathcal{A}$ as \textit{fully
reducible} if
\begin{itemize}
	\item If $G\in \mathcal{A}$ has a proper 3-edge cut $U$, then both $G^U$ and $G^{\overline{U}}$ are in $\mathcal{A}$. 
	\item The minimum cost Hamilton cycle of $G$ can be found in polynomial time for the
	graphs in $\mathcal{A}$ that do not have a proper 3-edge cut.
\end{itemize}
Cornu\'ejols showed that TSP can be solved in polynomial time for fully
reducible graphs~\cite{cornuejols1985traveling}.  An example of such a
fully reducible class are the Halin graphs~\cite{halin}.

For many fully reducible classes $\mathcal{A}$, they showed that for
$G\in \mathcal{A}$, if $G$ does not have a proper 3-edge cut, then the
convex hull of incidence vectors of Hamilton cycles of $G$ coincides
with a system of linear inequalities with polynomial
separation~\cite{cornuejols1985traveling}.  For example, they show if
$G=(V,E)$ does not contain any disjoint cycles, then $P_G=\{x\in
[0,1]^E\; :\; x(\delta(v))=2 \text{ for } v\in V\}$ is the convex hull
of incidence vectors of Hamilton cycles of $G$. Let us describe this
result in more detail. Suppose for any graph $G$ with no proper 3-edge
cuts that does not contain any disjoint cycles, we have
$\tsp(G)=P_G$. Now consider a graph $G=(V,E)$ with no disjoint cycles
that has a 3-edge cut $U$. In the graph $G^U$, we say the vertex
corresponding to the contracted set $U$ is a {\em pseudovertex}.
Suppose that the graphs $G^U$ and $G^{\overline U}$ contain no proper
3-edge cuts and suppose we can write $y$ restricted to the edge set of
each graph as a convex combination of Hamilton cycles of the
respective graph.  Let us consider the patterns around the
pseudovertices; if the edges adjacent to the pseudovertices are
$\{a,b,c\}$ then each vertex can be adjacent to two edges in a
Hamilton cycle and therefore, there are only three possible patterns
around a vertex: $\{\{a,b\},\{a,c\},\{b,c\}\}$.  Moreover, since each
pattern appears the same percentage of time (in the respective convex
combinations) for each pseudovertex, tours with corresponding patterns
can be {\em glued} over the 3-edge cut. In this case, the gluing
procedure is quite straightforward.  This reduction has also proven
quite useful for the minimum cost 2-edge-connected subgraph problem
(2EC) \cite{carrravi,boydlegault,philipe-masters}.

In contrast, it appears such a reduction is not known for TSP
tours. Indeed, gluing proofs cannot be easily extended to tours for
several reasons: (1) As just shown, they are often used for gluing
subgraphs (no doubled edges).  In TSP, we must allow edges to be
doubled, so there are too many possible patterns around a vertex.  For
example, if we allow each possible pattern corresponding to an even
degree, there are 13 possible patterns.  (2) Gluing tours over a
3-edge cut might result in disconnected Eulerian multigraphs.
Finally, (3) many of the algorithms based on gluing are not proven to
run in polynomial time~\cite{carrravi, boydlegault, philipe-masters}.

\subsection{Our Approach}

The main technical contribution of this paper is to show that for a
carefully chosen set of tours, we can design a gluing procedure over
{\em critical cuts}, which, roughly speaking are 
proper 3-edge cuts that are {\em
tight}: the $x$-values of the three edges crossing the cut sum to 2.
As we will see, a tour generated via a polyhedral version of
Christofides algorithm on a $\theta$-cyclic point can have eight
(rather than 13) possible patterns around a vertex.  While this is
still a lot, we identify certain conditions for the convex combination,
under which controlling the frequency of a single one of these
patterns allows us to control the frequency of the other seven.
Moreover, we show the frequency of one of the patterns (around an
arbitrary vertex $v$) depends on the fraction of times in the convex
combination of connectors that vertex $v$ is a leaf.  Thus, we can fix
a critical cut $U \subset V_n$ in $G_x$ and find a convex combination
of tours for $G_x^U$.  Then we can find a set of tours for
$G_x^{\overline{U}}$ such that the distribution of patterns around the
pseudovertex corresponding to $U$ matches that of the pseudovertex
corresponding to $\overline{U}$ in $G_x^U$, which enables us to glue
over critical cuts.

Applying this gluing procedure, we can reduce an instance of TSP on a
$\theta$-cyclic point to base-case instances, which contain no
critical cuts.  On a high level, our proof of Theorem \ref{main} for
such instances is based on Christofides algorithm: We show that a
$\theta$-cyclic point $x$ can be written as a convex combination of
connected, spanning subgraphs of $G_x$ with no doubled edges
(henceforth a {\em{connector of $G_x$}}) with certain properties and
then we show that the vector $z$, where $z_e = \frac{x_e}{2}$ for $e
\in H_x$ and $z_e = \frac{1}{2}-\frac{\theta}{10}$ for $e \in W_x$, can
be written as a convex combination of subgraphs, each of which can be
used for parity correction of a connector (henceforth a {\em{parity
    corrector}}).  Tight cuts are generally difficult to handle using
an approach based on Christofides algorithm, since
$(\frac{1}{2}-\epsilon)x$ is insufficient for parity correction of
a tight cut if it is crossed by an odd number of edges in the
connector.  However, in our base-case instances, there are only two
types of tight 3-edge cuts.  The first type of cut is a {\em
  degenerate tight cut}.  These cuts are easy to handle and we defer
their formal definition to Section \ref{sec:notation-cuts}.  The
second type of cut is a {\em vertex cut}, which we show are also easy
to handle.  In particular, the parity of vertex cuts can be addressed
with a key tool used by Boyd and Seb\H{o}~\cite{boydsebo} called
\textit{rainbow $v$-trees} (see Theorem \ref{boyd-sebo-rainbow}).
Using this in combination with a decomposition of the 1-edges into few
{\em induced matchings}, which have some additional required
properties, we can prove Theorem \ref{main} for the base case.
\nocite{IPbook}

\subsection{Organization}

In Section \ref{sec:notation}, we introduce notation and some
definitions relevant to cut structure in cyclic points.  We also review some
well-known polyhedral tools.  In Section \ref{sec:matching-patterns},
we present our main ideas and tools for gluing tours over critical cuts of
$G_x$, thereby reducing to TSP on base cases that only contain certain
types of tight cuts.  We then show in Section
\ref{sec:tour-base-case} that we are able to handle these remaining
tight cuts via an approach
similar---on a high-level---to Christofides algorithm.  For our gluing
approach to work, we need to choose the connectors to have certain
properties and to save on the 1-edges, we need to show that a vector
with less than half on each 1-edge belongs to the $O$-join polytope.
Both of these key technical ingredients can be found in Section
\ref{sec:tour-base-case}.  Sections \ref{sec:matching-patterns} and
\ref{sec:tour-base-case} contain a complete proof of Theorem
\ref{main}.  In Section \ref{sec:uniform-points}, we present an
application of Theorem \ref{main} to approximating TSP on
$\frac{2}{3}$-uniform points.  We also give an algorithm for TSP on
$\frac{2}{4}$-uniform points under certain assumptions.  Finally, in
Section \ref{sec:conclusion}, we make some concluding remarks and
present some problems for future research.

%% file: notation.tex
\section{Notation and Tools}\label{sec:notation}

Let $G=(V,E)$ be a graph. For a subset $U \subset V$ of vertices, let
$\delta_G(U)=\{uv\in E: u\in U, v\notin U\}$.  (We use $\delta(U) =
\delta_G(U)$ when the graph $G$ is clear from the context.)  Let $E[U]
=\{uv\in E: u\in U, v\in U\}$.

A multi-subset (henceforth \textit{multiset} for brevity) of edges of
$E$ is a set that can contain multiple copies of edges in $E$. A
multi-subgraph (henceforth \textit{multigraph} for brevity) of $G$ is
the graph on vertex set $V$ whose edge set is a multiset of $E$ (i.e.,
a multigraph can contain multiple copies of an edge). We sometimes
consider a multigraph $F$ of $G$ to be a multiset of edges of $G$.

The incidence vector of multigraph $F$ of $G$, denoted by $\chi^F$ is
a vector in $\mathbb{R}^E$ where $\chi^F_e$ is the multiplicity of $e$
in $F$. Let $F$ and $F'$ be two multigraphs of $G$, then $F+F'$ is the
multigraph that contains $\chi^F_e + \chi^{F'}_e$ copies of edge $e$
for $e\in E$. If $e$ is an edge in $F$, then $F-e$ is the multigraph
with incidence vector $\chi^{F-e}= \chi^F - \chi^{\{e\}}$.  For a
vector $x\in \mathbb{R}^E$ and a multigraph $F$ of $G$, we denote
$\sum_{e\in F}x_e\cdot \chi^F_e$ by $x(F)$. We use $\delta_F(U)$ to
refer to the multiset of edges in $F$ that have exactly one endpoint
in $U$. The degree of a vertex $v\in V$ in $F$ is the number of edges
in $F$ that are incident on $v$.

\subsection{Cuts in Cyclic Points}\label{sec:notation-cuts}

Let $x$ be a $\theta$-cyclic point with support graph $G_x = (V_n,
E_x)$.  The graph $G_x$ contains three types of tight cuts, by which
we mean a cut $U \subset V_n$ such that $x(\delta(U)) = 2$.
A {\em vertex cut} is a cut $U = \{u\}$.  Notice that for all $u \in
V_n$, $x(\delta(u)) = 2$.  The second type of cut is a {\em critical cut}.

\begin{definition}
A proper cut $U \subset V_n$ in $G_x$ is called a {\em critical cut} if
$|\delta(U)| = 3$ and $\delta(U)$ contains exactly one edge $e$ with
$x_e = 1$.  Moreover, for each pair of edges in $\delta(U)$, their
endpoints in $U$ (and in $V\setminus{U}$) are distinct.
\end{definition}

We refer to the third type of cut 
as a {\em degenerate tight cut}.

\begin{definition}
A proper cut $U \subset V_n$ in $G_x$ is called a {\em degenerate
  tight cut} if $|\delta(U)| = 3$, $|U| > 3$ and $|V\setminus{U}| > 3$
and the two fractional edges in $\delta(U)$ share an endpoint in
either $U$ or $V\setminus{U}$.  
\end{definition}
For a degenerate tight cut $U$, let
$\delta(U)=\{e,f,g\}$, such that $f$ and $g$ are the fractional edges
that share an endpoint $v$. Let $e_v$ be the unique 1-edge incident on
$v$. Observe that $\{e,e_v\}$ forms a 2-edge cut in $G_x$.

Let $x$ be a $\theta$-cyclic point and $U\subset V_n$ be a cut in
$G_x$. Also, let $\overline{U} = V_n\setminus U$. We can obtain
$\theta$-cyclic point $x^U$ by contracting set $\overline{U}$ in $G_x$
to a single vertex.  (Respectively, we can obtain a $\theta$-cyclic
point $x^{\overline{U}}$ by contracting set $U$ in $G_x$ to a single
vertex.) 
We let $v_{\overline{U}}$ denote 
the vertex corresponding to set $\overline{U}$ in $\Gux$
(respectively, $v_U$ corresponds to $U$ in $\Goux$).

\begin{observation}\label{obs:cornercut}
Let $U\subseteq V_n$ be a minimal critical
cut in $G_x$ (i.e., for $S\subset U$, the cut defined by $S$ is not
critical in $G_x$). Then, $\Gux$ does not contain any critical cuts.
\end{observation}
\begin{proof}
Suppose for contradiction that there is $S \subset V(\Gux)$ that is a
critical cut of $\Gux$. We can assume that $v_{\overline{U}}\notin
S$. Moreover, $S\subsetneq U$. This is a contradiction to minimality
of $U$ since $S$ constitutes a critical cut in $G_x$ as well.
\end{proof}

\begin{observation}\label{obs:inductivecriticalcut}
Suppose that $G_x$ has $k$ critical cuts. Let $U$ be a critical cut of
$G_x$. The number of critical cuts in $\Gux$ is at most $k-1$.
\end{observation}
\begin{proof}
Clearly, $U$ is not a critical cut of $\Gux$. We show that there is correspondence between the critical cuts of $\Gux$ and $G_x$. This implies that $\Gux$ can have at most $k-1$ critical cuts. If $S$ is a critical cut of $\Gux$ we can assume without loss of generality that $v_{\overline{U}}\notin S$. Hence, $S$ is also a critical cut of $G_x$. 
\end{proof}

\subsection{Polyhedral Basics}

Let $G=(V,E)$ and let $x$ be a vector in $\mathbb{R}^E$.  Consider a
collection of multigraphs $\mathcal{F}$ of $G$. We say $\lambda=
\{\lambda_F\}_{F\in \mathcal{F}}$ are convex multipliers for
$\mathcal{F}$ if $\sum_{F\in\mathcal{F}} \lambda_F = 1$ and $\lambda_F
\geq 0$ for $F\in \mathcal{F}$.  We say {\emph
  {$\{\lambda,\mathcal{F}\}$ is a convex combination for $x$}} if
$\lambda= \{\lambda_F\}_{F\in \mathcal{F}}$ are the convex multipliers
for $\mathcal{F}$ and $x= \sum_{F\in\mathcal{F}}\lambda_F$. We say
{\emph{$x$ can be written as a convex combination of multigraphs in
    $\mathcal{F}$}} if we can find such $\mathcal{F}$ and $\lambda$ in
polynomial time in the size of $x$. Here by the size of $x$ we refer
to $|E|$ (i.e., the number of edges in the support of $x$).

\subsubsection{The $v$-Tree Polytope and Rainbow $v$-Trees}
Let $G=(V,E)$ be a graph. For a vertex $v\in V$, a $v$-tree is a subgraph $F$ of $G$ such that $|F\cap \delta(v)|=2$ and $F\setminus \delta(v)$ induces a spanning tree of $V\setminus \{v\}$. Denote by $v\vtree(G)$ the convex hull of incidence vectors of $v$-trees of $G$. The $v\vtree(G)$ is characterized by the following linear inequalities.
\begin{align}
v\vtree(G)=& \{x\in[0,1]^E: x(\delta(v))=2,\nonumber\\ & \; x(E[U])\leq
|U|-1 \mbox{ for all $\emptyset\subset U\subseteq V\setminus \{v\}$},\; x(E)=|V|\}.
\end{align} 

\begin{observation}\label{obs:subtinvtree}
	We have $\subt(K_n)\subseteq v\vtree(K_n)$ for all $v\in K_n$.
\end{observation}

\begin{observation}\label{lem:in-subtour}
Let $x\in \subt(K_n)$ be such that $G_x$ is 3-edge-connected and cubic. Let $\mathcal{C}$ be any 2-factor in $G_x$. Define vector $y$ to have $y_e =\frac{1}{2}$ for $e\in \mathcal{C}$, $y_e = 1$ for $e\in E_x\setminus \{C\}$ and $y_e=0$ otherwise. Then $y\in \subt(K_n)\subseteq v\vtree(K_n)$. 
\end{observation}
\begin{proof}
	Take $\emptyset\subset U \subset V_n$. If $U=\{v\}$ for some $v\in V_n$, then $x(\delta(U)= 2$ as $\mathcal{C}\cap \delta(v) = 2$.  
	
	If $|\delta(U)|\geq 4$, then clearly $x(\delta(U))\geq 2$.  Otherwise, $|\delta(U)| =
	3$.  Since at most two edges in $\delta(U)$ belong to
	$\mathcal{C}$, there is at least one edge $e\in \delta(U)$
	with $x_e=1$. Hence, $x(\delta(U)) \geq 2$. Therefore, $x \in
	\subt(K_n)$. We have $x\in v\vtree(K_n)$ by Observation \ref{obs:subtinvtree}.
\end{proof}

It can be deduced from the discussion above that a vector $x$ in the
subtour elimination relaxation can be written as a convex combination
of $v$-trees for any vertex $v$ in $G_x$. In fact, the $v$-trees in
this convex combination can satisfy some additional properties.

\begin{definition}
	Let $G=(V,E)$ and $v$ be a vertex of $G$. Let $\mathcal{P}$ be
        a collection of disjoint subsets of $E$. A $\mathcal{P}$-rainbow $v$-tree, namely $T$, is a $v$-tree of $G$ such that $|T\cap P|=1$ for $P\in \mathcal{P}$.
\end{definition}

The following theorem can be proved via the matroid intersection theorem \cite{Edmonds2003} and Observation \ref{obs:subtinvtree}. 

\begin{thm}[\cite{BLI},\cite{boydsebo}]\label{boyd-sebo-rainbow}
	Let $x\in \subt(K_n)$ and $\mathcal{P}$ be a collection of disjoint
	subsets of $E_x$ such that $x(P)= 1$ for $P\in \mathcal{P}$. Then $x$
	can be written as a convex combination of $\mathcal{P}$-rainbow
	$v$-trees of $K_n$ for any $v\in V_n$.
\end{thm}
Gr{\"o}tchel and Padberg \cite{padberg} observed that $v$-trees of a connected graph $G=(V,E)$ satisfy the basis axioms of a matroid. For $x\in \subt(K_n)$ we have $x\in v\vtree(K_n)$ by Observation \ref{obs:subtinvtree}. Also, $\mathcal{P}$ defines a partition matroid where each base intersect each part of $\mathcal{P}$ exactly once. Therefore, vector $x$ is in the convex hull of incidence vector of common basis of the partition matroid defined by $\mathcal{P}$ and the matroid whose basis are the $v$-trees of $K_n$.

\subsubsection{The $O$-join Polytope}\label{sec:o-join-polytope}

Let $G=(V,E)$ be a graph and $O\subseteq V$ where $|O|$ is even. An $O$-join of $G$ is a subgraph $J$ of $G$ where a vertex $v\in V$ has odd degree in $J$ if and only if $v\in O$. Let $O\join(G)$ be the convex hull of incidence vectors of $O$-joins of $G$. Edmonds and Johnson \cite{tjoin} showed the following description for the $O\join(G)$.
\begin{align}
O\join(G) = \{ z\in [0,1]^E\;: \;& z(\delta(U)\setminus A)-z(A)\geq 1-|A|\label{o-join-exact}\\
&\quad\quad\text{ for } U\subseteq V, A\subseteq \delta(U), |U\cap O|+|A| \text{ odd}\}.\nonumber
\end{align}

\begin{observation}\label{obs:tjoinbasic}
If $x\in \subt(K_n)$, then $\frac{x}{2}\in O\join(K_n)$ for any $O\subseteq V_n$ with $|O|$ even. 
\end{observation}
\begin{proof}
Let $z=\frac{x}{2}$. For $U\subseteq V_n$, we have $z(\delta(U))\geq 1$. Moreover, for $e\in E_n$, we have $z_e \leq \frac{1}{2}$ since $x_e\leq 1$. Therefore, for $A\subseteq \delta(U)$ we have $z(A) \leq \frac{|A|}{2}$. This implies $z(\delta(U))-2z(A) \geq 1- |A|$. 
\end{proof}

The following observation states a useful property of 
$O$-joins that can be obtained via \eqref{o-join-exact}.

\begin{observation}\label{humbleobservation}
	Let $G=(V,E)$ be a graph, and let $O\subseteq V$ be a subset
	of vertices such that $|O|$ is even.  Let $z\in O\join(G)$,
	and $z(\delta(u))\leq 1$ for all $u\in V$. Then $z$ can be
	written as convex combination of $O$-joins  of $G$ denoted by $\{\psi,\cal{J}\}$ such that for $u\in O$ we have
	$|J\cap \delta(u)|=1$ for $J\in \cal{J}$. 
\end{observation} 

\begin{proof}
	By \cite{tjoin}, if $z\in O\join(G)$, then $z$ can be
	written as a convex combination of $O$-joins of $G$ denoted by $\{\psi,\cal{J}\}$. Let
	$u$ be a vertex in $O$. We have $z(\delta(u))\leq 1$. On the
	other hand, for every $J\in \mathcal{J}$ we have
	$|J\cap \delta(u)| \geq 1$. Therefore, $|J\cap \delta(u)|
	=1$.  \end{proof}

\subsubsection{Proof of Theorem \ref{3/2}: Polyhedral Analysis of Christofides}\label{proofofchris}

Now, we are ready to prove Theorem \ref{3/2}. 
\christofides*
\begin{proof}
Let $x\in \subt(K_n)$.  Then by Observation \ref{obs:subtinvtree},
$x\in v\vtree(K_n)$ for some $v\in V_n$. Hence, we can find $v$-trees
$\mathcal{T}$ and convex multipliers $\lambda$ for $\mathcal{T}$ such
that $x = \sum_{T\in \mathcal{T}}\lambda_T\chi^{T}$. For each
$T\in\mathcal{T}$, let $O_T$ be the set of odd degree vertices of
$T$. Notice that $\frac{x}{2} \in O_T\join(K_n)$ for all $T\in
\mathcal{T}$.  This implies that $\frac{x}{2}$ can be written as a
convex combination of $O_T$-joins $\mathcal{J}^T$ of $K_n$ with convex
multipliers $\theta= \{\theta_J\}_{J\in \mathcal{J}^T}$ (i.e.,
$\frac{x}{2} = \sum_{J \in \mathcal{J}^T} \theta_J \chi^J$).  Notice
that for $T\in \mathcal{T}$ and $J\in \mathcal{J}^T$, multigraph $T+J$
is a tour of $K_n$.
Hence, $\sum_{T\in \mathcal{T}}\lambda_T\sum_{J\in
  \mathcal{J}^T}\theta_J \chi^{T+J}\in \tsp(K_n)$.
Moreover, 
\begin{eqnarray*}
\sum_{T\in \mathcal{T}}\lambda_T\sum_{J\in
  \mathcal{J}^T}\theta_J \chi^{T+J}& = &
\sum_{T\in \mathcal{T}}\lambda_T (\chi^T + \sum_{J\in
  \mathcal{J}^T}\theta_J \chi^{J})\\& = & 
\sum_{T\in \mathcal{T}}\lambda_T \chi^T + \sum_{T \in \mathcal{T}}\lambda_T\sum_{J\in
  \mathcal{J}^T}\theta_J \chi^{J}\\ & = & 
x + \sum_{T \in \mathcal{T}} \lambda_T \frac{x}{2}\\ 
& = & x + \frac{x}{2} \sum_{T \in \mathcal{T}} \lambda_T\\ &  = & \frac{3}{2}x. 
\end{eqnarray*}
 Therefore,
$\frac{3}{2}x\in \tsp(K_n)$.
\end{proof}

%% file: main.tex
\section{Matching Patterns: Gluing Tours Over Critical Cuts}\label{sec:matching-patterns}

Let $x$ be a $\theta$-cyclic point and $G_x=(V_n,E_x)$ be the support
of $x$.  In this section, we present an approach for gluing tours over
critical cuts of $G_x$.  One property of the tours we construct,
which is crucial to enable this gluing procedure, is that every tour
contains at least one copy of each 1-edge. 
This allows us to assume that
$x$ belongs to a subclass of $\theta$-cyclic points in which (i) $G_x$
is cubic, (ii) $W_x$, the 1-edges of $G_x$, form a perfect matching,
and (iii) $H_x$, the fractional edges of $G_x$, form a 2-factor.  
We can make this assumption, because we can contract
a path of 1-edges to a single 1-edge; the tour of the new cubic graph
yields a tour for the original subcubic graph.
We
work under this assumption throughout this section and in Section
\ref{sec:tour-base-case}.

For a vertex $u\in V$, denote by $e_u$ the unique 1-edge in
$G_x$ that is incident on $u$.  Let $\delta(u)= \{e_u,\fu,\gu\}$ where
$\fu$ and $\gu$ are the two fractional edges incident on $u$ and
$x_{\fu} = \theta$ and $x_{\gu} = 1-\theta$.  In each tour of $G_x$, a
multiset of edges from $\delta(u)$ belongs to the tour.  We call this
multiset {\em{the pattern around $u$}}.  Denote by ${\mathbb{P}}_u$
the set of possible patterns around a vertex $u$ in a tour of $G_x$
that contain at least one copy of the 1-edge $e_u$ and in which $u$
has degree either two or four (see Figure \ref{fig:patterns}).
\begin{equation*}
{{\mathbb{P}}}_u = \{ \{2e_u\},
\{e_u, \fu\}, \{e_u,\gu\}, \{2e_u,2\fu\}, \{2e_u,2\gu\} ,\{2e_u,\fu,\gu\},
\{e_u,2\fu,\gu\}, \{e_u,\fu,2\gu\}\}.
\end{equation*}
In every tour we construct, the pattern around each vertex $u \in V$
will be some pattern from $\mathbb{P}_u$.  There are other multisets
of $\delta(u)$ that can be valid patterns around $u$ in a tour.  For
example, the pattern $\{\fu,\gu\}$ could be the pattern around $u$ in
some tour.  However, in our construction this pattern will never be
the pattern around $u$ as we always include at least one copy of $e_u$
in a tour.  Formally, we write $\delta_F(u) = p$ if the pattern
around $u$ in the tour $F$ is $p \in \mathbb{P}_u$.

\input{GluingFigure}

\begin{definition}
A tour $F$ of $G_x=(V_n,E_x)$ is a {\em{handpicked tour of $G_x$}} if
for all $u\in V_n$, the pattern around $u$ in $F$ belongs to 
$\mathbb{P}_u$ (i.e., if $\delta_F(u) \in \mathbb{P}_u$ for all
$u \in V_n$). 
\end{definition}

If $\mathscr{F}$ is a set of tours and $\phi = \{\phi_F\}_{F \in
  \mathscr{F}}$ is a set of convex multipliers, then we use $\phi(p)=
\sum_{F\in \mathscr{F}: \delta_F(u)= p}\phi_F$ to denote the {\em
  pattern frequency} of $p \in \mathbb{P}_u$ in the convex combination
$\{\phi, \mathscr{F}\}$.  Notice that $\phi(p) \in [0,1]$.  Moreover,
if $\mathscr{F}$ is a set of handpicked tour of $G_x$, then for all $u \in V_n$, we
have $\sum_{p \in \mathbb{P}_u} \phi(p) = 1$.
For each $u \in V_n$, we define the {\em pattern profile} of vertex
$u$ in the convex combination $\{\phi, \mathscr{F}\}$ to be the eight
values $\{\phi(p)\}$ for all $p \in \mathbb{P}_u$. 

Another key parameter of a convex combination is the frequency of
doubled edges.  For the convex combination $\{\phi, \mathscr{F}\}$,
define $\phi_2(e) = \sum_{F\in \mathscr{F}: \chi^F_e = 2}\phi_F$ for
all $e \in E_x$.  The pattern profile of a vertex $u$ turns out to be
directly related to the occurence of doubled edges from $\delta(u)$.
This dependence is formalized in the next observation, which states that for
a convex combination of handpicked tours, if the parameters
$\phi_2(e)$ are fixed for $e \in \delta(u)$, then the pattern profile
for each vertex $u$ depends only on the pattern frequency of the
pattern $\{2e_u\}$.

\begin{observation}\label{magic-lemma}
Let $y$ and $q$ be vectors in $\mathbb{R}^{E_x}_{\geq 0}$. Suppose $y$ can
be written as a convex combination of tours of $G_x$ denoted by $\{\phi,\mathscr{F}\}$ such that
for all $e \in E_x$, we have $\phi_2(e) = q_e$.

Then for each vertex $u\in V_n$ and for each pattern $p \in
\mathbb{P}_u$, the frequency of pattern $p$, $\phi(p)$, in this convex
combination
is uniquely determined by the frequency of pattern $\{2e_u\}$.
 \end{observation}

\begin{proof}
Suppose $\phi(\{2e_u\})=
\zeta_u$ for some $\zeta_u \in [0,1]$.  Then for each $u \in V_n$, the
following identities hold with respect to the convex combination
$\{\phi, \mathscr{F}\}$.
	\begin{align*}
	\sum_{p\in \mathbb{P}_u: \chi^p_{e} = 2}\phi(p)  &= q_{e} & \text{ for } e\in \{e_u,\fu,\gu\},\\
	\sum_{p\in \mathbb{P}_u: \chi^p_{e} = 1}\phi(p)  &= y_e - 2q_{e}
        & \text{ for } e\in \{e_u,\fu,\gu\},\\
	\sum_{p \in \mathbb{P}_u} \phi(p) &= 1,\\		\phi(\{2e_u\}) &= \zeta_u.
	\end{align*}
Since the above system of equations has full rank (i.e.,
the variables are $\phi(p)$ for $p \in \mathbb{P}_u$), it has a unique
solution.  Therefore, $\phi(p)$ is a function of $\zeta_u$ for all
$p\in \mathbb{P}_u$. \end{proof}

We apply Observation \ref{magic-lemma} to control the pattern profile
of a pseudovertex $u$ by constructing tours in which the pattern
frequency of $\{2e_u\}$ can be set arbitrarily.  This enables us to
prove Theorem \ref{main} with an inductive (gluing) approach.  For
such an approach to work, we need to prove a stronger statement.  Let
$\alpha$ be a constant in $(0,1]$ that we will fix later.

\begin{proposition}\label{main-induction}
Define $y \in \mathbb{R}^{E_x}_{\geq 0}$ as follows: $y_e =
\frac{3}{2}-\frac{\alpha\theta}{2}$ for $e \in W_x$ and $y_e =
\frac{3}{2}x_e$ for $e \in H_x$. Then $y$ can be written as a convex
combination of handpicked tours of $G_x$ denoted by
$\{\phi,\mathscr{F}\}$ such that
\begin{enumerate}

\item[(i)] $\phi_2(e) = \frac{1}{2}-\frac{\alpha\theta}{2}$, for $e\in W_x$, and

\item[(ii)] $\phi_2(e) = \frac{x_e^2}{2}$ for all $e\in H_x$.
\end{enumerate}
\end{proposition}

Observe that Proposition \ref{main-induction} implies Theorem
\ref{main}.  As mentioned previously, it is in fact stronger; we
construct tours for the base cases ($\theta$-cyclic points whose support graphs
have no critical cuts) and the additional properties in Proposition
\ref{main-induction} enable us to ``glue'' these tours together over
the critical cuts of $G_x$.  Hence, our induction is on the number of
critical cuts in $G_x$.

Observation \ref{gluenow} gives sufficient conditions under which we
can glue tours of $G_x^{U}$ and $G_x^{\overline{U}}$ together over the critical
cut to obtain tours for $G_x$ that preserve key properties.

\begin{observation}\label{gluenow}
Let $U \subset V_n$ be a critical cut of $G_x$.  Suppose $x^U$ can be
written as a convex combination of handpicked tours of $\Gux$ denoted
by $\{\phi^U,\mathscr{F}^U\}$ with the following properties.  (And
suppose the same holds for $x^{\overline{U}}$, $\Goux$
$\{\phi^{\overline{U}},\mathscr{F}^{\overline{U}}\}$, respectively.)
\begin{enumerate}

\item[(i)] The pattern profiles of vertices $v_U$ and $v_{\overline{U}}$
  are the same in their respective convex combinations. 

\item[(ii)] For every tour $F \in \mathscr{F}^U$,
  $F\setminus{\delta(v_{\overline{U}})}$ induces a connected multigraph on
  $U$.  

\end{enumerate}
Then $x$ can be written as a combination of handpicked tour of $G_x$ denoted by $\{\phi,\mathscr{F}\}$ such that
\begin{enumerate}
\item[(a)] $\phi_2(e) = \phi^U_2(e)$ for $e \in E(\Gux)$, 
\item[(b)] $\phi_2(e) = \phi^{\overline{U}}_2(e)$ for $e \in
  E(\Goux)$, and
 \item[(c)] $|\mathscr{F}|\leq |\mathscr{F}^U|+ |\mathscr{F}^{\overline{U}}|$.
\end{enumerate}
\end{observation}

\begin{proof}
First, we prove the following simple claim.

\begin{claim}\label{gluenow2}
	Consider a graph $G=(V,E)$ and nonempty $U\subset V$ such that
        $U$ is a 3-edge-cut in $G=(V,E)$.  Let $F_U$ be a tour in
        $G^U$ and let $F_{\overline{U}}$ be a tour in
        $G^{\overline{U}}$ such that $\chi^{F_U}_e =
        \chi^{F_{\overline{U}}}_e$ for $e\in \delta(U)$.  Moreover,
        assume that $F_U\setminus{\delta(v_{\overline{U}})}$ induces a
        connected multigraph on $U$.  Then the multiset of edges $F$
        defined as $\chi^{F}_e = \chi^{F_U}_e$ for $e\in E(G^U)$ and
        $\chi^{F}_e = \chi^{F_{\overline{U}}}_e$ for $e\in
        E(G^{\overline{U}})$ is a tour of $G$.
\end{claim}

\begin{cproof}
	It is clear that $F$ induces an Eulerian spanning multigraph on $G$,
	but we need to ensure that $F$ is connected.  For example, the
	tour induced on $F_{\overline{U}}\setminus{\delta(v_U)}$ might not
	be connected.  However, since the subgraph of $F_U$ induced on
	the vertex set $U$ is connected, the tour $F$ is connected:
	each vertex in $\overline{U}$ is connected to some vertex in $U$.
\end{cproof}

We observe that if the pattern profiles of $v_U$ and
$v_{\overline{U}}$ with respect to the convex combinations $\{\phi^U,
\mathscr{F}^U\}$ and $\{\phi^{\overline{U}}, \mathscr{F}^{\oU}\}$,
respectively, are the same, then we can always find two tours $F_U \in
\mathscr{F}$ and $F_{\overline{U}} \in \mathscr{F}^{\oU}$ such
that the pattern around $v_{\overline{U}}$ in $F_U$ is the same as the
pattern around $v_U$ in $F_{\overline{U}}$.  We can apply Claim
\ref{gluenow2} to obtain a new tour $F$, to which we assign convex
multiplier $\phi_F = \min\{\phi^U_{F_U},
\phi^{\oU}_{F_{\overline{U}}}\}$ and add to set $\mathscr{F}$.
Then we subtract $\phi_F$ from each of these convex multipliers,
remove tours with convex multipliers zero from $\mathscr{F}^U$ and
$\mathscr{F}^{\oU}$, and repeat. Observe the total number of tours in
$\mathscr{F}$ is at most $|\mathscr{F}^U|
+ |\mathscr{F}^{\oU}|$.

We need to show that each tour $F \in \mathscr{F}$ is handpicked.
This follows from the fact that the pattern around each vertex in $U$
in $F$ is the same as the pattern around it in $F_U$.  Moreover, each
edge $e \in E_x \cap \delta(U)$ is doubled in a tour $F$ of $G_x$ iff
it is doubled in both $F_U$ and in $F_{\overline{U}}$.  For $e \in
E(\Gux)\setminus{\delta(U)}$, edge $e$ is doubled iff it is doubled in
$F_U$.  Analogously, each vertex in $\overline{U}$ has the same
pattern in $F$ as it has in $F_{\overline{U}}$, and each edge $e \in
E(\Goux)\setminus{\delta(U)}$ is doubled iff it is doubled in
$F_{\overline{U}}$.  Thus, properties (a) and (b) hold for the convex
combination $\{\phi, \mathscr{F}\}$.
\end{proof}

In the base case (where graphs $G_x = (V_n,E_x)$ have no critical
cuts), each tour in the convex combination that we construct consists
of a connector plus a parity correction (as we will describe in Section
\ref{sec:tour-base-case}).  For each $u \in V_n$, the previously
introduced (and yet to be fixed) parameter $\alpha$ is a lower bound
on the fraction of connectors in this convex combination in which $u$
has degree two.  Let the parameter $\eta_u$ denote the (exact)
fraction of connectors in which vertex $u$ has degree one.  (Note that
$\eta_u$ can be different for every vertex.)  In our construction of a
convex combination of tours $\{\phi, \mathscr{F}\}$ for the base case,
it will be the case that the frequency of pattern $\{2e_u\}$ will be
equal to $\frac{\eta_u}{2}$.  Thus, controlling the value of $\eta_u$
allows us to control $\phi(\{2e_u\})$.  A key technical tool is that
when we construct a convex combination of tours $\{\phi,
\mathscr{F}\}$ for graph $G_x$ with no critical cuts (i.e., a base
case graph), we can always ensure that for one arbitrarily chosen
vertex $v$, the value of $\eta_v$ (and hence $\phi(\{2e_v\})$) can be
chosen arbitrarily.  This allows us to ensure that the pattern
frequency of $\{2e_{v_{\overline{U}}}\})$ equals the pattern frequency
of $\{2e_{v_U}\})$ in the given convex combination for $\Goux$ (i.e.,
Condition (i) in Observation \ref{gluenow}) when $\Gux$ is a base
case.

\begin{lemma}\label{main-base-case}
Suppose $G_x$ contains no critical cuts.
Fix any vertex $v \in V$
and fix constant $\zeta$ with $0 \leq \zeta \leq \frac{(1-\alpha)\theta}{2}$. 
Define
$y \in \mathbb{R}^{E_x}_{\geq 0}$ as follows:
$y_e = \frac{3}{2}-\frac{\alpha\theta}{2}$ for $e \in W_x$ and $y_e =
\frac{3}{2}x_e$ for $e \in H_x$. Then $y$ can be written as a convex combination of  handpicked tours of $G_x$ denoted by
$\{\phi,\mathscr{F}\}$ with the following properties.
\begin{enumerate}
\item[(i)] $\phi_2(e) = \frac{1}{2}-\frac{\alpha\theta}{2}$, for $e\in
  W_x$, 
\item[(ii)] $\phi_2(e) = \frac{x_e^2}{2}$ for all $e\in H_x$, 
\item[(iii)] $\phi(\{2e_v\}) = \zeta$, and
\item[(iv)] $F\setminus{\delta_F(v)}$ induces a connected multigraph on $V
  \setminus{v}$ for each $F \in \mathscr{F}$.
\end{enumerate}
	\end{lemma}

Notice that Lemma \ref{main-base-case} implies Proposition
\ref{main-induction} for $\theta$-cyclic points whose support graphs
have no critical cuts. We prove Lemma \ref{main-base-case} in the next
section. In the remainder of this section, we show how Lemma
\ref{main-base-case} implies Proposition \ref{main-induction}.

\begin{proof}[Proof of Proposition \ref{main-induction}]
	Suppose $G_x$ contains $t$ critical cuts. We prove the
        statement by induction on $t$. In fact we show that the
        running time of our algorithm is polynomial in $n$ and $t$. To this end, we show that the convex combination in our
construction contains at most $btn^d$ trees, where $b$ and $d$ are
constants. Notice that $t$ itself is a polynomial bounded by $n$.
	
	If $t=0$, then $G_x$ does not contain a critical cut, then apply Lemma
        \ref{main-base-case}.  Otherwise, find the minimal critical
        cut $U$ of $G_x$. By Observations \ref{obs:cornercut} and
        \ref{obs:inductivecriticalcut}, graph $\Gux$ does not contain any critical cuts and $\Goux$ contains at most $t-1$ critical cuts. 
        
        Define $y^{\overline{U}}$ as follows: $y^{\overline{U}}_e=
        \frac{3}{2}-\frac{\alpha\theta}{2}$ for $e\in
        W_{x^{\overline{U}}}$ and $y^{\overline{U}}_e=
        \frac{x^{\overline{U}}}{2}$ for $e\in
        H_{x^{\overline{U}}}$. We apply the induction hypothesis on
        $\Goux$ to write $y^{\overline{U}}$ as convex combination of
        handpicked tours of $\Goux$ denoted by
        $\{\phi^{\overline{U}},\mathscr{F}^{\overline{U}}\}$ such that
        (i) $\phi^{\overline{U}}_2(e) =
        \frac{1}{2}-\frac{\alpha\theta}{2}$ for $e\in
        W_{x^{\overline{U}}}$, and (ii) $\phi_2(e) =
        \frac{(x^{\overline{U}}_e)^2}{2}$ for all $e\in
        H_{x^{\overline{U}}}$. By induction
        $|\mathscr{F}^{\overline{U}}|\leq b(t-1)(|\overline{U}|+1)^d$.
        		
    	Let $\zeta^*=\phi^{\overline{U}}(\{e_{v_U}\})$. Define
        $y^{U}_e=\frac{3}{2}-\frac{\alpha\theta}{2}$ for $e\in
        W_{x^U}$ and $y^U_e = \frac{x^U_e}{2}$ for $e\in
        H_{x^U}$. Applying Lemma \ref{main-base-case} to
        $\theta$-cyclic point $x^U$ with $\zeta =\zeta^*$ we can write
        $y^U$ as convex combination of handpicked tours of $\Gux$
        denoted by $\{\phi^U,\mathscr{F}^U\}$ such that (i)
        $\phi^U_2(e) = \frac{1}{2}-\frac{\alpha\theta}{2}$ for $e\in
        W_{x^U}$, (ii) $\phi^U_2(e) = \frac{(x^U_e)^2}{2}$ for all
        $e\in H_{x^U}$, (iii) $\phi^U(\{2e_v\}) = \zeta^*$, and (iv)
        $F\setminus{\delta_F(v_{\overline{U}})}$ induces a connected
        multigraph on $U$ for each $F \in \mathscr{F}^U$. Moreover
        $|\mathscr{F}^U|\leq b(|U|+1)^d$.

    	Now we apply Observation \ref{gluenow} to write $y$ as the desired convex combination. This convex combination contains at most $|\mathscr{F}^U|+|\mathscr{F}^{\overline{U}}| \leq b(|U|+1)^d +b(t-1)(|\overline{U}|+1)^d$. Note that for $d\geq 2$ this number is at most $btn^d$, as $|U|+|\overline{U}|=n$ , $|U|\geq 3$, and $|\overline{U}|\geq 3$. 
\end{proof}

%% file: GluingFigure.tex
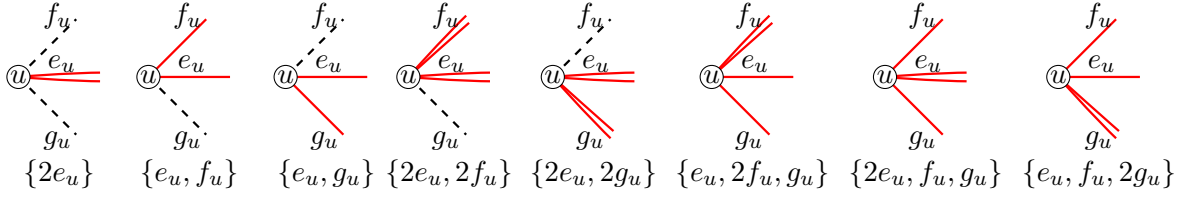
\begin{figure}[h]
	\centering
	\begin{tikzpicture}[scale=0.32]
	\begin{scope}[scale=0.6]
	\draw [-] [red, line width=0.3mm] plot [smooth, tension=1] coordinates {(0,0) (2.82,0.2) (5.65,0.3)};
	\draw [-] [red, line width=0.3mm] plot [smooth, tension=1] coordinates {(0,0) (2.82,-0.2)(5.65,-0.3)};		
	\draw [dashed] [black, line width=0.3mm] plot [smooth, tension=0] coordinates {(0,0) (4,4)};
	\draw [dashed] [black, line width=0.3mm] plot [smooth, tension=0] coordinates {(0,0) (4,-4)};	
	\draw[black,fill=white] (0,0) ellipse (0.8 cm  and 0.8 cm);
	\node (u) at (0,0) {$u$};
	\node (eu) at (3,0.9) {$e_u$};
	\node (f) at (2.7,4.25) {$f_u$};
	\node (g) at (2.7,-4.4) {$g_u$};
	\node (title) at (2.82,-6.7) {$\{2e_u\}$};
	\end{scope}
	\begin{scope}[scale=0.6,xshift=9cm]
	\draw [-] [red, line width=0.3mm] plot [smooth, tension=0] coordinates {(0,0) (5.65,0)};
	\draw [-] [red, line width=0.3mm] plot [smooth, tension=0] coordinates {(0,0) (4,4)};
	\draw [dashed] [black, line width=0.3mm] plot [smooth, tension=0] coordinates {(0,0) (4,-4)};	
	\draw[black,fill=white] (0,0) ellipse (0.8 cm  and 0.8 cm);
	\node (u) at (0,0) {$u$};
	\node (eu) at (3,0.9) {$e_u$};
	\node (f) at (2.7,4.25) {$f_u$};
	\node (g) at (2.7,-4.4) {$g_u$};
	\node (title) at (2.82,-6.7) {$\{e_u,f_u\}$};
	\end{scope}
	\begin{scope}[scale=0.6,xshift=18.5cm]
	
	\draw [-] [red, line width=0.3mm] plot [smooth, tension=0] coordinates {(0,0) (5.65,0)};
	\draw [dashed] [black, line width=0.3mm] plot [smooth, tension=0] coordinates {(0,0) (4,4)};
	\draw [-] [red, line width=0.3mm] plot [smooth, tension=0] coordinates {(0,0) (4,-4)};	
	\draw[black,fill=white] (0,0) ellipse (0.8 cm  and 0.8 cm);
	\node (u) at (0,0) {$u$};
	\node (eu) at (3,0.9) {$e_u$};
	\node (f) at (2.7,4.25) {$f_u$};
	\node (g) at (2.7,-4.4) {$g_u$};
	\node (title) at (2.82,-6.7) {$\{e_u,g_u\}$};
	\end{scope}
	\begin{scope}[scale=0.6,xshift=37cm,yshift=-0cm]
	\draw [-] [red, line width=0.3mm] plot [smooth, tension=1] coordinates {(0,0) (2.82, 0.2) (5.65,0.3)};
	\draw [-] [red, line width=0.3mm] plot [smooth, tension=1] coordinates {(0,0) (2.82,-0.2) (5.65,-0.3)};		
	\draw [dashed] [black, line width=0.3mm] plot [smooth, tension=0] coordinates {(0,0) (4,4)};
	\draw [-] [red, line width=0.3mm] plot [smooth, tension=1] coordinates {(0,0) (2,-1.8) (4.2,-3.75)};
	\draw [-] [red, line width=0.3mm] plot [smooth, tension=1] coordinates {(0,0) (2,-2.2) (4,-4.25)};	
	\draw[black,fill=white] (0,0) ellipse (0.8 cm  and 0.8 cm);
	\node (u) at (0,0) {$u$};
	\node (eu) at (3,0.9) {$e_u$};
	\node (f) at (2.6,4.25) {$f_u$};
	\node (g) at (2.6,-4.4) {$g_u$};
	\node (title) at (2.82,-6.7) {$\{2e_u,2g_u\}$};
	\end{scope}
	\begin{scope}[scale=0.6,xshift=27cm]
	\draw [-] [red, line width=0.3mm] plot [smooth, tension=1] coordinates {(0,0) (2.82, 0.2) (5.65,0.3)};
	\draw [-] [red, line width=0.3mm] plot [smooth, tension=1] coordinates {(0,0) (2.82,-0.2) (5.65,-0.3)};		
	\draw [-] [red, line width=0.3mm] plot [smooth, tension=1] coordinates {(0,0) (2,2.2) (4,4.25)};
	\draw [-] [red, line width=0.3mm] plot [smooth, tension=1] coordinates {(0,0) (2,1.8) (4.2,3.75)};
	\draw [dashed] [black, line width=0.3mm] plot [smooth, tension=0] coordinates {(0,0) (4,-4)};	
	\draw[black,fill=white] (0,0) ellipse (0.8 cm  and 0.8 cm);
	\node (u) at (0,0) {$u$};
	\node (eu) at (3,0.9) {$e_u$};
	\node (f) at (2.6,4.25) {$f_u$};
	\node (g) at (2.6,-4.4) {$g_u$};
	\node (title) at (2.82,-6.7) {$\small\{2e_u,2f_u\}$};
	\end{scope}
	\begin{scope}[scale=0.6,xshift=60cm,,yshift=0cm]
	\draw [-] [red, line width=0.3mm] plot [smooth, tension=1] coordinates {(0,0) (2.82, 0.2) (5.65,0.3)};
	\draw [-] [red, line width=0.3mm] plot [smooth, tension=1] coordinates {(0,0) (2.82,-0.2) (5.65,-0.3)};		
	\draw [-] [red, line width=0.3mm] plot [smooth, tension=0] coordinates {(0,0) (4,4)};
	\draw [-] [red, line width=0.3mm] plot [smooth, tension=0] coordinates {(0,0)  (4,-4)};	
	\draw[black,fill=white] (0,0) ellipse (0.8 cm  and 0.8 cm);
	\node (u) at (0,0) {$u$};
	\node (eu) at (3,0.9) {$e_u$};
	\node (f) at (2.6,4.25) {$f_u$};
	\node (g) at (2.6,-4.4) {$g_u$};
	\node (title) at (2.82,-6.7) {$\{2e_u,f_u,g_u\}$};
	\end{scope}
	\begin{scope}[scale=0.6,xshift=72cm,yshift=0cm]
	\draw [-] [red, line width=0.3mm] plot [smooth, tension=0] coordinates {(0,0)  (5.65,0)};
	\draw [-] [red, line width=0.3mm] plot [smooth, tension=0] coordinates {(0,0)  (4,4)};
	\draw [-] [red, line width=0.3mm] plot [smooth, tension=1] coordinates {(0,0) (2,-1.8) (4.2,-3.75)};
	\draw [-] [red, line width=0.3mm] plot [smooth, tension=1] coordinates {(0,0) (2,-2.2) (4,-4.25)};	
	\draw[black,fill=white] (0,0) ellipse (0.8 cm  and 0.8 cm);
	\node (u) at (0,0) {$u$};
	\node (eu) at (3,0.9) {$e_u$};
	\node (f) at (2.8,4.25) {$f_u$};
	\node (g) at (2.8,-4.4) {$g_u$};
	\node (title) at (2.82,-6.7) {$\{e_u,f_u,2g_u\}$};
	\end{scope}
	\begin{scope}[scale=0.6,xshift=48cm,yshift=0]
	\draw [-] [red, line width=0.3mm] plot [smooth, tension=0] coordinates {(0,0) (5.65,0)};
	\draw [-] [red, line width=0.3mm] plot [smooth, tension=1] coordinates {(0,0) (2,2.2) (4,4.25)};
	\draw [-] [red, line width=0.3mm] plot [smooth, tension=1] coordinates {(0,0) (2,1.8) (4.2,3.75)};
	\draw [-] [red, line width=0.3mm] plot [smooth, tension=0] coordinates {(0,0)  (4,-4)};	
	\draw[black,fill=white] (0,0) ellipse (0.8 cm  and 0.8 cm);
	\node (u) at (0,0) {$u$};
	\node (eu) at (3,0.9) {$e_u$};
	\node (f) at (2.6,4.25) {$f_u$};
	\node (g) 	at (2.6,-4.4) {$g_u$};
	\node (title) at (2.82,-6.7) {$\{e_u,2f_u,g_u\}$};
	\end{scope}
	\end{tikzpicture}
	\caption{The different patterns in $\mathbb{P}_u$. The red solid edges are in the tour and black dashed edges are not used in the handpicked tour.}
	\label{fig:patterns}
\end{figure}

%% file: base-case.tex
\section{Finding Tours in the Base Case: Proof of
Lemma \ref{main-base-case}}\label{sec:tour-base-case}

In this section we present the proof of Lemma \ref{main-base-case}.
We fix $G_x = (V_n, E_x)$ to be the support graph of a $\theta$-cyclic
point $x$.  In the base case, $G_x$ contains no critical cuts.  (The
lemmas in this section will be applied to this base case, but some
hold even when $G_x$ is not a base case.)  Moreover, we remind the
reader that we assume the support of $G_x$ is cubic.  (See discussion
in the beginning of Section
\ref{sec:matching-patterns}.)

A key tool we will use is
that the 1-edges of $G_x$ can be partitioned into five induced
matchings in $G_x$.  A set $M \subset W_x$ is an {\em induced matching
  of $G_x$} if $M$ is a vertex induced subgraph of $G_x$ and $M$ is a
matching.  For each induced matching $M$, we find a set of connectors
$\mathcal{T}$ of $G_x$ where for each 1-edge $e$ in $M$, both
endpoints of $e$ have degree two in every $T\in \mathcal{T}$.  Thus,
when $G_x$ has no critical cuts, a 1-edge $e$ in $M$ does not belong
to any odd cuts in $T$ that are tight cuts in $G_x$.  For each 1-edge
$e$ in $M$, we can therefore reduce usage of $e$ in the parity
correction from $\frac{1}{2}$ to $\frac{1-\theta}{2}$; each 1-edge
saves $\frac{\theta}{2}$ exactly $\frac{1}{5}$ of the times. This
yields the saving of $\frac{\theta}{10}$ on the 1-edges as stated in
Lemma \ref{main-base-case} with $\alpha =\frac{1}{5}$.

The induced matchings require some additional properties that we need
for technical reasons as we will see later.  Recall that for a vertex
$u$ in $G_x$ we denote by $e_u$ the unique 1-edge incident on $u$.
Let $\N(u)$ denote the two vertices that are the other endpoints of
the fractional-edges incident on $u$. In other words, suppose
$\delta(u) = \{e_u, f_u, g_u\}$ and suppose that $w_1$ and $w_2$ are
the other endpoints of $f_u$ and $g_u$, respectively. Then $\N(u)
= \{w_1, w_2\}$.  The proof of Lemma \ref{alpha} is deferred to
Section \ref{ind-match}.

\begin{restatable}{lemma}{mainInducedMatching}
	\label{alpha}
Suppose $G_x$ has no critical cuts.  Let $v$ be a vertex in $V_n$ and
let $\N(v) = \{w_1, w_2\}$. The set of $1$-edges in $G_x$, $W_x$, can
be partitioned into five induced matchings $\{M_1, \dots, M_5\}$ such
that for $i\in [5]$, the following properties hold.  \begin{enumerate}
		
		\item[(i)] $|M_i \cap \{e_v, e_{w_1}, e_{w_2} \}| \leq 1$.
		
		\item[(ii)] For $U\subseteq V_n$ such that $|\delta(U)|=3$, 
		$|\delta(U)\cap M_i|\leq 1$.
		
		\item[(iii)] For $U\subseteq V_n$ such that $|\delta(U)|=2$, 
		$|\delta(U)\cap M_i|$ is even.

	\end{enumerate}
\end{restatable}

For the rest of this section, let $v$ be a fixed vertex in $V_n$,
$N_2(v) =\{w_1,w_2\}$ and
let $\{M_1, \dots, M_h\}$ denote the partition of $W_x$ into induced
matchings with the additional properties enumerated in
Lemma \ref{alpha}.  These properties will be used to ensure that we
can save on the edges in $W_x$ when augmenting connectors of $G_x$
(with parity correctors) into tours.  While Lemma \ref{alpha} implies
that $h=5$, we will use $h$ throughout this section, since if
Lemma \ref{alpha} could be proved with, say, four matchings, it would
allow a larger value of $\alpha = \frac{1}{h}$ and hence directly
yield a better bound in the statement of Lemma \ref{main-base-case}.

The proof of Lemma \ref{main-base-case} consists of two main
parts. First we show there is a convex combination of connectors of
$G_x$ that satisfy certain properties.  Second, we show that for each
connector, we can find parity correctors such that the union of a
connector and a parity corrector is a tour.  

\subsection{Constructing Connectors}

Now we will show how to construct connectors for $G_x$.  We will apply
this when $G_x$ is a base case (i.e., $G_x$ has no critical cuts), but
Definition \ref{propertyP} and Lemmas \ref{lambda=0}
and \ref{arbitlambda} apply even when this is not the case.  Recall
that $v$ is fixed vertex in $V_n$ and that $N_2(v) = \{w_1,w_2\}$.

\begin{definition}\label{propertyP}
Suppose $M \subset W_x$ is a subset of 1-edges of $G_x$. Let $\Lambda$
be a constant such that $0\leq \Lambda \leq \theta$.  Suppose $x$ can
be written as a convex combination of connectors of $G_x$ denoted by
$\{\lambda ,\mathcal{T}\}$. Then we say $P(v,M,\Lambda)$ holds for
$\{\lambda, \mathcal{T}\}$ if it has the following properties.
	\begin{enumerate}
		\item $\sum_{T \in \mathcal{T}:
			|\delta_{T}(v)|=1}\lambda_T                  =\sum_{T
			\in \mathcal{T}: |\delta_{T}(v)|=3}\lambda_T=
		\Lambda$ and $\sum_{T \in \mathcal{T}:
			|\delta_{T}(v)|=2}\lambda_T=1-2\Lambda$.
		
		\item For each edge $st \in M$,
		$|\delta_{T}(s)|=|\delta_{T}(t)|=2$ for all $T \in \mathcal{T}$.
		\item $T\setminus \delta_{T}(v)$ induces a connector on $V\setminus \{v\}$.
	\end{enumerate}
\end{definition}

Let us explain why the properties described above are useful in our
construction. The first property allows us to control the fraction of
time vertex $v$ has degree one in a connector in the convex combination
$\{\lambda,\mathcal{T}\}$, which in turn will allow us to control the
fraction of time a tour has the pattern $\{2e_v\}$ around $v$.
This flexibility is required to perform the gluing procedure;
it allows us to manipulate the convex combination of connectors to
have the desired pattern profile for the pseudovertex (which will be
$v$). The second condition ensures that no 1-edge in $M$ is part of a
tight cut that is crossed an odd number of times in a connector
$T\in\mathcal{T}$. Lastly, the third property 
guarantees that we maintain connectivity of the tours after gluing
them together over critical cuts.

We defer the proofs of the next two lemmas to Section \ref{sec:spanning}.
\begin{restatable}{lemma}{lambdazero}
	\label{lambda=0}
Suppose $M \subset W_x$ forms an induced matching in $G_x$ and edge
$e_v \in M$.  Then $x$ can be written as a convex combination of
connectors of $G_x$ denoted by $\{\lambda,\mathcal{T}\}$ for which
$P(v,M,0)$ holds.
\end{restatable}

\begin{restatable}{lemma}{arbitlambda}
	\label{arbitlambda}
Let $\Lambda$ be any constant such that $0 \leq \Lambda \leq \theta$.
Suppose $M \subset W_x$ forms an induced matching in $G_x$, $e_v
\notin M$ and $|M\cap \{e_{w_1},e_{w_2}\}| \leq 1$.  Then $x$ can be
written as a convex combination of connectors of $G_x$ denoted by
$\{\lambda,\mathcal{T} \}$ for which $P(v,M,\Lambda)$ holds.
\end{restatable}

Recall that $\{M_1,\ldots,M_h\}$ is the partition of $W_x$ into
induced matchings obtained via Lemma \ref{alpha}.  Assume without loss
of generality that $e_v\in M_1$. For $i = 1$, let $\mathcal{T}_1$ be a
set of connectors of $G_x$ and let $\{\vartheta, \mathcal{T}_1\}$ be a
convex combination for $x$ for which $P(v,M_1,0)$ holds (by
Lemma \ref{lambda=0}).  For $i \in \{2,
\dots, h\}$, let $\mathcal{T}_i$ be a set of connectors of $G_x$ and
let $\{\vartheta, \mathcal{T}_i\}$ be a convex combination for $x$ for
which $P(v,M_i,\frac{\Lambda}{1-\alpha})$ holds (by Lemma
\ref{arbitlambda}).  Notice that $\frac{\Lambda}{1-\alpha}\leq \theta$
since $\Lambda \leq (1-\alpha)\theta$.

We can write $x$ as a convex combination of connectors from
$\mathscr{T}$, by weighting each set $\mathcal{T}_i$ by $\alpha$.  In
particular, we have $x= \alpha \sum_{i
=1}^h \sum_{T \in \mathcal{T}_i} \vartheta_T \chi^{T}$.  For each
$T \in \mathscr{T}$, let $\sigma_T = \alpha \cdot \vartheta_T$.  Then
$\{\sigma, \mathscr{T}\}$ is a convex combination for $x$. Observe
that since $x_e = 1$ for $e\in W_x$, we have $W_x \subseteq T$ for
$T\in \mathscr{T}$. From Definition
\ref{propertyP} and Lemmas \ref{lambda=0} and \ref{arbitlambda}, we
observe the following.

\begin{claim}\label{firstclaiminmainproof}
	For each $T \in \mathscr{T}$, $T \setminus{\delta(v)}$ induces a
	connected, spanning subgraph on $V\setminus{\{v\}}$.
\end{claim}

\subsection{Constructing Parity Correctors}

For each $T\in \mathscr{T}$, let $O_T$ be the set of odd degree
vertices of $T$. In the second part of the proof we show that each
connected subgraph $T\in\mathscr{T}$ has a ``cheap'' convex
combination of $O_T$-joins.
\begin{restatable}{lemma}{tjoinLemma}
	\label{tjoin}
Suppose $G_x$ has no critical cuts.  Let $M \subset W_x$ be a
        subset of 1-edges of $G_x$ such that each 3-edge cut in $G_x$
        contains at most one edge from $M$.  Let $O \subseteq V$ be a
        subset of vertices such that $|O|$ is even and for all
        $e=st\in M$, neither $s$ nor $t$ is in $O$. Also suppose for
        any set $U\subseteq V$ such that $|\delta(U)|=2$, both $|U\cap
        O|$ and $|\delta(U)\cap M|$ are even.  Define vector $z$ as
        follows: $z_e = \frac{1}{2}$ if $e \in W_x$ and $e\notin M$,
        $z_e = \frac{1-\theta}{2}$ if $e \in M$, and $z_e =
        \frac{x_e}{2}$ if $e \in H_x$.  Then vector $z\in
        O\join(G_x)$.
\end{restatable}

For each $i \in [h]$, define $z_e^{i} = \frac{1-\theta}{2}$ if $e \in
{M_i}$ and $z_e^{i} = \frac{x_e}{2}$ otherwise.  For each $T \in
\mathcal{T}_i$, let $O_T \subseteq V$ be the set of odd-degree
vertices of $T$.  By construction, we have $V(M_i)\cap O_T=\emptyset$.
By Lemma \ref{tjoin}, we have $z^i\in O_T\join(G_x)$. So by
Observation \ref{humbleobservation} we can write $z^i$ as a convex
combination of $O_T$-joins of $G_x$ denoted by $\{\psi^T,\cal{J}_T\}$
where $|J\cap \delta(u)|=1$ for $u\in O_T$ and $J\in \cal{J}_T$. This
implies that $x+z^i$ can be written as a convex combination of tours
of $G_x$.  We denote this set of tours by $\mathcal{F}_{i}$ and we let
$\mathscr{F} = \{F\in \mathcal{F}_i \; : \; i\in \{1, \dots, h\}\}$.
Now for $F\in \mathscr{F}$ we have $F=T+J$ for some $T\in \mathscr{T}$ and $J\in \mathcal{J}_T$. Define $\phi_F = \sigma_T\cdot \psi_J^T$. The vector $\sum_{i=1}^{h}\alpha(x+z^i)$ can be written as $\{\phi,\mathscr{F}\}$.
\begin{claim}
		Every tour in $\mathscr{F}$ is handpicked.
	\end{claim}
\begin{cproof}
	Let $F\in \mathscr{F}$. By construction, $F= T+ {J}$,
        where $T\in \mathscr{T}$ and $J\in \mathcal{J}_T$. Let $u\in
        V_n$ and $\delta(u) = \{e_u,f_u,g_u\}$. Notice that
        $\chi^F_{e_u}\geq 1$, since $e_u \in T$ for all
        $T\in \mathscr{T}$. Hence, we only need to show that
$|\delta_F(u)| < 6$.        
Suppose for contradiction that
$|\delta_F(u)|\geq 6$. This implies that
$|\delta_T(u)| = 3$ and $|\delta_J(u)|=3$. 
However, if
$|\delta_T(u)|=3$, then $u\in O_T$. 
From Observation \ref{humbleobservation},
if $u\in O_T$ and $z(\delta(u)) \leq 1$, then
$|J \cap \delta(u)|=1$ which is a contradiction.
\end{cproof}

\begin{claim}\label{xmidclaim}
Suppose $G_x$ contains no
	critical cuts.  Define vector $y \in \mathbb{R}^{E_x}$ as $y_e =
	\frac{3}{2}-\frac{\alpha\theta}{2}$ for $e\in W_x$ and $y_e = \frac{3}{2}x_e$
	for $e\in H_x$.  Then $\{\phi, \mathscr{F}\}$ is a convex combination
	for $y$.
\end{claim}
\begin{cproof}
	We need to show that $y = \sum_{i=1}^{h}\alpha(x + z^i)$.  First, let
	$e$ be a 1-edge of $G_x$ and $M_j$ be the induced matching that
	contains $e$. Then, $x_e = 1$, $z^i_e = \frac{1}{2}$ for $i \in
	[h]\setminus \{j\}$ and $z^j_e = \frac{1-\theta}{2}$. Hence,
	\begin{equation*}
	\sum_{i=1}^{h} \alpha(x_e+z_e^i) = \sum_{\ell=1}^{h}
	\alpha \cdot \frac{3}{2} - \alpha\cdot  \frac{\theta}{2}  =
	\frac{3}{2} -\frac{\alpha\theta}{2}.
	\end{equation*}
	For a fractional $e$ of $G_x$, we have 
	$z^i_e=\frac{x_e}{2}$ for $i\in [h]$, so $\sum_{i=1}^{h}\alpha(x_e+z^i_e
	)=\frac{3}{2}x_e$.
\end{cproof}

Now we prove some additional useful properties of the convex
combination $\{\phi, \mathscr{F}\}$ for $\theta$-cyclic point $x$. 

\begin{claim}\label{edgedoubling-frequency}
For convex combination $\{\phi, \mathscr{F}\}$, we have
$\phi_2(e) = \frac{1}{2} - \frac{\alpha \theta}{2}$ for $e \in W_x$
and $\phi_2(e) = \frac{(x_e)^2}{2}$.
\end{claim}

\begin{cproof}
Notice that $\phi$ defines a probability distribution on
$\mathscr{F}$. We sample $F$ from $\mathscr{F}$ with the probabilities
defined by $\phi$. 
Recall that $\phi_2(e) = \Pr[e \text{ is doubled in $F$}]$.
Moreover, recall that $F=T+J$ where $T\in \mathscr{T}$ and
$J\in \mathcal{J}_T$ and $T$ is associated with some matching
$M\in \{M_1,\ldots,M_h\}$ (i.e., $T \in \mathcal{T}_i$ for $i \in [h]\}$). For $e\in E_x$, we have \begin{equation*}
\Pr[e_u \text{ is doubled in $F$ }]   = \Pr[e_u \in J~\text{and}~ e_u \in T ]=\Pr[e_u\in T]\cdot \Pr[e_u \in J].
\end{equation*}
For $e\in H_x$, this implies that $\phi_2(e) = x_e\cdot \frac{x_e}{2} = \frac{(x_e)^2}{2}$. Also 
\begin{align*}
\Pr[e_u \text{ is doubled in $F$ }] &=\Pr[e_u \in J] \\ 
&=\Pr[e_u\in J| e_u\in M]  \cdot
\Pr[e_u \in M]+\Pr[e_u\in J| e_u\notin M]  \cdot
\Pr[e_u \notin M]\\
 &=\frac{1-\theta}{2}\cdot  \alpha + \frac{1}{2} \cdot (1-\alpha)\\
&= \frac{1}{2} - \frac{\alpha \theta}{2}.
\end{align*}

\end{cproof}

Claims \ref{firstclaiminmainproof}, \ref{xmidclaim} and \ref{edgedoubling-frequency} yield Lemma \ref{main-base-case}. It remains to prove Lemmas \ref{lambda=0}, \ref{arbitlambda} and \ref{tjoin}.

\subsection{Proofs of Lemmas \ref{lambda=0} and \ref{arbitlambda} for
Constructing Connectors}\label{sec:spanning}

In this section, we prove Lemmas \ref{lambda=0} and \ref{arbitlambda},
which are necessary in order to write a $\theta$-cyclic point $x$ as a
convex combination of connectors with property $P$ described in
Definition \ref{propertyP}.

\lambdazero*
\begin{proof}
	For each $st\in M$, pair the half-edges incident on $s$ and pair those
	incident on $t$ to obtain disjoint subsets of edges $\mathcal{P}$. Decompose $x$ into a convex combination of $\mathcal{P}$-rainbow
	$v$-trees $\mathcal{T}$ (i.e., $x=\sum_{T \in
		\mathcal{T}}\lambda_T\chi^{T}$) via Theorem \ref{boyd-sebo-rainbow}. This is
	the desired convex combination since for all $T \in \mathcal{T}$, we
	have $|\delta_{T}(v)| = 2$ and $|\delta_{T}(u)|=2$ for all endpoints
	$u$ of edges in $M$.  Thus, the first and second conditions are
	satisfied. The third condition holds by definition of $v$-trees.
\end{proof}

\arbitlambda*
\begin{proof}
	As in the proof of Lemma \ref{lambda=0}, for each $st\in M$, pair the
	half-edges incident on $s$ and pair those incident on $t$ to obtain a
	collection of disjoint subsets of edges $\mathcal{P}$. Apply Theorem
	\ref{boyd-sebo-rainbow} to obtain $\{\lambda, \TT\}$ which is a convex
	combination for $x$, where $\TT$ is a set of $\mathcal{P}$-rainbow
	$v$-trees (i.e., $x=\sum_{T \in \TT}\lambda_T\chi^{T}$). Notice that
	this convex combination clearly satisfies the second requirement in
	Definition \ref{propertyP}.
	
	Now let $\delta(v) = \{e_v,f,g\}$, where $w_1$ and $w_2$ are the other
	endpoints of $f$ and $g$, respectively.  Assume $x_f = \theta$
        and $x_g = 1-\theta$.
Since $x= \sum_{T \in \TT}\lambda_T
	\chi^{T}$ and $x_{e_v} = 1$, we have $e_v\in T$ for $T \in \TT$.  In
	addition, we have $|\delta_{T}(v)|=2$ for all $T \in \TT$ by the
	definition of $v$-trees.  Hence, $\sum_{T \in \TT: f\in T, g\notin
		T}\lambda_T=\theta$ and
$\sum_{T\in \TT: f\notin T, g\in T}\lambda_T=
	x_f=1-\theta$.  Define 
$$\TT_f = \{T \in \TT : f\in T \text{ and } g\notin T\} \text{ and }
        \TT_g = \{T \in \TT : g\in T \text{ and } f\notin T\},$$ where
        $\TT_f \cup \TT_g = \TT$ and $\TT_f \cap \TT_g = \emptyset$.
        We can also assume that there are subsets $\TT^1_f \subseteq
        \TT_f$ and $\TT^1_g \subseteq \TT_g$ such that $\sum_{T \in
          \TT^1_f}\lambda_T=\Lambda$ and $\sum_{T \in
          \TT^1_g}\lambda_T= \Lambda$, since $\Lambda\leq \theta$.
        Now we consider two cases
\begin{enumerate}
\item	If $e_{w_1}\notin M$:
For $T \in \TT_f^1$,
	replace $T$ with $T-f$.  For $T \in \TT_g^1$, replace $T$
	with $T + f$.  

\item If $e_{w_1} \in M$:
For $T \in \TT_f^1$,
	replace $T$ with $T+g$.  For $T \in \TT_g^1$, replace $T$
	with $T -g$.  
\end{enumerate}

For all $T \in \TT\setminus{( \TT_f^1 \cup \TT_g^1)}$, keep $T$ as is.
Observe that $T\in \TT$ is still a connector of
$G_x$: for every $T \in \TT_f$, $T-f$ is a spanning tree of $G_x$ and
for every $T \in \TT_g$, $T-g$ is spanning tree of $G_x$.  Observe
that in the first case, the degree of $e_{w_2}$ is preserved for every
$T$, and in the second case the degree of $e_{w_1}$ is preserved for
every $T$.  We want to show that the new convex combination
$\{\lambda, \TT\}$ is the desired convex combination for $x$.  Notice
that in the first case,
	\begin{align*}\sum_{T \in \TT}\lambda_T\chi^{T}_f &=
	\sum_{T \in \TT_f^1}\lambda_T\chi^{T}_f
	+\sum_{T \in \TT_f\setminus{\TT_f^1}}\lambda_T\chi^{T}_f
	+\sum_{T \in \TT_g^1}\lambda_T\chi^{T}_f + \sum_{T \in \TT_g \setminus{\TT_g^1}}\lambda_T\chi^{T}_f \\
	&=  0 + (\theta -\Lambda) + \Lambda + 0= x_f.
	\end{align*}
In the second case,
	\begin{align*}\sum_{T \in \TT}\lambda_T\chi^{T}_g &=
	\sum_{T \in \TT_f^1}\lambda_T\chi^{T}_g
	+\sum_{T \in \TT_f\setminus{\TT_f^1}}\lambda_T\chi^{T}_g
	+\sum_{T \in \TT_g^1}\lambda_T\chi^{T}_g + \sum_{T \in \TT_g \setminus{\TT_g^1}}\lambda_T\chi^{T}_f \\
	&=  \Lambda + 0 + 0 + (1 - \theta -\Lambda) = x_g.
	\end{align*}

	So $x=\sum_{T \in \TT}\lambda_T\chi^{T}$.  Moreover, notice
        that for $T \in \TT$, $T\setminus \delta_{T}(v)$ still induces
        a connector on $V\setminus \{v\}$ since we did not
        remove any edge in $T\setminus{\delta(v)}$ from the $v$-tree
        $T$.  Finally, for each vertex $s$ with $e_s\in M$, we have
        $|\delta_{T}(s)|=2$ for all $T \in \TT$. To observe this,
        notice that the initial convex combination satisfies this
        property for vertex $s$ (since the convex combination is
        obtained via Theorem \ref{boyd-sebo-rainbow}).  In the
        transformation of the convex combination we only change edges
        incident on $w_1$ and $w_2$, so if $s\neq w_1,w_2$ the
        property clearly still holds after the transformation.  If $s
        \in \{w_1, w_2\}$, then as noted previously, we do not remove
        or add an edge incident on $s$ if $e_s \in M$.
\end{proof}

\subsection{Proof of Lemma \ref{tjoin} for Constructing Parity Correctors}\label{sec:parity}

We use $O$-joins as parity correctors for each $T \in \mathscr{T}$.
We now give the complete proof of Lemma \ref{tjoin}.

\begin{proof}[Proof of Lemma \ref{tjoin}]
Our goal is to show that $z$ belongs to $O\join(G_x)$.  By definition,
	$z\in [0,1]^{E_x}$.  Now we will show that $z$ satisfies the
	constraint \eqref{o-join-exact}.  First, we state three useful
	claims.

\begin{claim}\label{z-large}
If $z(\delta(U))\geq 1$ for $U \subset V_n$,
	then $z(\delta(U)\setminus A)-z(A) \geq 1 - |A|$.
\end{claim}

\begin{cproof}
We have $z(\delta(U)\setminus A)-z(A) = z(\delta(U))-2z(A)$.  
Since $z_e \leq \frac{1}{2}$ for all $e\in E_x$, we have $
z(\delta(U))-2z(A) \geq 1 - |A|$.
\end{cproof}

\begin{claim}\label{no-matching-edges}
If $\delta(U)\cap M =\emptyset$, we have $z(\delta(U))\geq 1$.
\end{claim}

\begin{cproof}
This follows from the fact that 
for every edge $e\notin M$, we have $z_e
	= \frac{x_e}{2}$. 
\end{cproof}

\begin{claim}\label{same-parity}
For all $U \subset V_n$, $|\delta(U) \cap W_x|$
	and $|\delta(U)|$ have same parity.
\end{claim}

\begin{cproof}
This follows from the fact that $|\delta(U) \cap H_x|$ is always
even since $H_x$ is a 2-factor of $G_x$.
\end{cproof}

We consider the following cases. \textbf{Case 1: } $|\delta(U)\cap W_x|\geq 3$,  \textbf{Case 2: }  $|\delta(U)\cap W_x|=2$, and \textbf{Case 3: }  $|\delta(U)\cap W_x|=1$.
	
	 \paragraph{Case 1: } If $|\delta(U)\cap W_x|\geq 4$, then
	 $z(\delta(U))\geq 2-2\theta \geq 1$. Thus we may assume
	 $|\delta(U)\cap W_x|=3$. If $|\delta(U)\geq 4$, then
	 $z(\delta(U))\geq \frac{3}{2} - \frac{3}{2}\theta
	 + \frac{\theta}{2} \geq 1$, since
	 $\theta \leq \frac{1}{2}$. If $|\delta(U)|=3$, then by
	 assumption we have $|\delta(U)\cap M|\leq 1$. Thus,
	 $z(\delta(U))\geq 1$. Thus, Claim \ref{z-large} applies in
	 each subcase.

	\paragraph{Case 2: } In this case, if $|\delta(U)|\geq 4$,
        then $z(\delta(U))\geq (1-\theta) + \theta \geq 1$.  By
        Claim \ref{same-parity}, the only remaining subcase to
        consider is when $|\delta(U)|=2$.
	
 By assumption, $|U\cap O|$ is even. Hence, $|A|$ must be odd, which
	implies that $|A|=1$.  Let $\delta(U) = \{e',e''\}$. Since
	$|\delta(U)|=2$, we have either $|\delta(U)\cap M| = 2$ or
	$|\delta(U)\cap M| = 0$.  In both cases $z_{e'} =
	z_{e''}$. Hence, $z(A) = z(\delta(U)\setminus A)$. Therefore,
	$z(\delta(U)\setminus A)-z(A) = 0 = 1-|A|$.

	\paragraph{Case 3: } If $\delta(U) \cap M = \emptyset$, then
	by Claim \ref{no-matching-edges} we have $z(\delta(U))\geq
	1$. Hence, we assume $|\delta(U)\cap M|=1$.  By
	Claim \ref{same-parity}, we only need to consider 
the following
	cases: \textit{Case 3i: } $|\delta(U)|=3$, and \textit{Case
	3ii: } $|\delta(U)|\geq 5$.
	
	\begin{itemize}
		\item[] \textit{Case 3i:} Notice that
		$x(\delta(U))\geq 2$.  In this case, $\delta(U)$ is either
		a critical cut, a vertex cut, or a degenerate tight cut. We assumed that $G_x$
		has no critical cuts. So $U$ is either a vertex cut or a degenerate tight cut. We prove in both cases that $|U\cap O|$ is even. Then we only need to
		consider $|A|$ odd. If $|A|=1$, then
		$z(\delta(U)\setminus A) - z(A) \geq \frac{\theta}{2} \geq 0=
		1-|A|$. 
		If $|A|=3$, then $z(\delta(U)\setminus A) - 		z(A) \geq -1 + \frac{\theta}{2} \geq -2= 1-|A|$. 
		
		 If $U=\{u\}$ for some $u\in
		V_n$, then $u\notin O$ by assumption. Otherwise, $U$
                is a degenerate tight cut. Let
                $\delta(U)=\{e_u,f_v,g_v\}$ where
                $\{f_v,g_v\}=\delta(v)\cap H_x$ for some $v\in
                V_n$. Notice that $\delta(U\setminus \{v\})=
                \{e_u,e_v\}$ and $e_u\in M$. This implies by
                assumption that $e_v \in M$, which implies that $v
                \notin O$.  Since
 $|(U\setminus \{v\})\cap O|$ is even, hence
                $|U\cap O|$ is even. 
		
		\item[] \textit{Case 3ii:} 
Since $|\delta(U)|\geq 5$, if there is an edge
		$e\in \delta(U)\cap H_x$ with $z_e
		= \frac{1-\theta}{2}$, then $z(\delta(U))\geq
		1- \theta + 3\cdot \frac{\theta}{2} = 1 + \frac{\theta}{2}$. Therefore, $\theta<\frac{1}{2}$ and for all edges $e$ in $\delta(U)\cap H_x$, we have $z_e = \frac{\theta}{2}$. 
		
		Let $\mathcal{C}$ be the collection of cycles in
		$H_x$. Since $\theta<\frac{1}{2}$, every cycle in
		$\mathcal{C}$ is even length. Clearly, any cut crosses
		every cycle $C\in \mathcal{C}$ an even number of
		times. If all the edges in $\delta(U)\cap C$ have the
		same $x$ value, then $|U\cap V(C)|$ is even. We have 
		\begin{equation}
		|U| = \sum_{C\in \mathcal{C}: V(C)\subseteq U} |V(C)\cap U| + \sum_{C\in \mathcal{C}: V(C)\cap U\neq \emptyset} |V(C)\cap U|.
		\end{equation}
		By the argument above and the fact that $|V(C)|$ is
		even for all $C\in \mathcal{C}$, we conclude that
		$|U|$ is even. Since $G_x$ is a cubic graph, this
		implies that $|\delta(U)|$ is even. However, by
		Claim \ref{same-parity}, we know that $|\delta(U)|$ is odd.
	\end{itemize}
This concludes the case analysis and the proof.
\end{proof}	
\subsection{Proof of Lemma \ref{alpha}: Partitioning 1-edges into Induced Matchings}\label{ind-match}

The goal of this section is to prove the following lemma.
\mainInducedMatching*		

We say $\delta(U)$ is a triangular 3-cut if $|U|= 3$ or $|V\setminus
U|=3$, and $|\delta(U)|=3$.  A bad 3-edge cut is a proper 3-edge cut
that is not triangular.  We construct the desired partition of $W_x$
into induced matchings by gluing over the bad cuts of $G_x$ and
perform induction on the number of bad 3-edge cuts.  We prove Lemma
\ref{alpha} using a two-phase induction.  Claim \ref{alpha1} is the
base case and Claims \ref{alpha2} and \ref{alpha3} are the first and
second inductive steps.

\begin{claim}\label{alpha1}
	Suppose $G_x$ is 3-edge-connected and contains no bad 3-edge
	cuts. Then Lemma \ref{alpha} holds.
\end{claim}
\begin{cproof}
	In $G_x$, contract every edge in $W_x$. We get a connected 4-regular graph $H=(W_x,H_x)$.  An independent set in $H$ corresponds
	to a set of edges in $W_x$ that forms an induced matching in $G_x$.
	We consider two cases. If $H$ is the complete graph on five vertices,
	then partition the vertex set into five independent sets, which
	corresponds to five induced matchings in $G_x$.  Notice that the
	condition (i) from Lemma \ref{alpha} is satisfied since each induced
	matching contains one edge.
	
	If $H$ is not the complete graph on five vertices, by Brook's
	Theorem (see Theorem 8.4 in \cite{bondy}) we can partition the
	vertices of $H$ into four independent sets where each
	independent set corresponds to an induced matching
	$\{M_1, \ldots, M_4\}$ in $G_x$ and these four induced
	matchings partition $W_x$.  If
	$|M_i\cap \{e_r,e_{w_1},e_{w_2}\}|\leq 1$ for $i \in [4]$,
	then we are done.  Otherwise, assume without loss of
	generality that $\{e_{w_1},e_{w_2}\}\in M_4$.  Then let $M'_4
	= M_4\setminus \{e_{w_1}\}$.  The desired partition is
	$\{M_1,M_2,M_3,M'_4,\{e_{w_1}\}\}$. Thus, condition (i) is
	satisfied.
	
	Now we prove condition (ii).
	First, consider a vertex $u \in V$ and the cut $\delta(u)$ in
	$G_x$.  Clearly $|\delta(u)\cap M_i|\leq |\delta(u)\cap W_x|\leq
	1$.  For a triangular 3-cut, $\delta(U)=\{e_1,e_2,e_3\}$, we cannot have
	$|\delta(U)\cap \{e_1,e_2,e_3\}|\geq 2$, since $\delta(U)\subseteq W_x$ and
	no pair of edges from $\delta(U)$ can belong to an induced matching.
	Since condition (iii) does not apply, this completes the proof of the claim.
\end{cproof}

\begin{claim}\label{alpha2}
	Suppose $G_x$ is 3-edge-connected. Then Lemma \ref{alpha}
	holds.
\end{claim}
\begin{cproof}
	Now let us consider a bad cut.  In particular, consider graph $G_x$
	with 3-edge-cut $\delta(U)=\{e_{1}, e_{2}, e_{3}\}$, and assume
	without loss of generality that $r\in U$.  Let $s_1, s_2$ and $s_3$
	be the endpoints of $e_{1}, e_{2}$ and $e_{3}$ that are in $U$, and
	$t_1,t_2$ and $t_3$ be the other endpoints.  Notice that $s_1,s_2,s_3$
	(and analogously $t_1,t_2,t_3$) are distinct vertices since $G_x$ is
	3-edge-connected.  Construct graph $G_1 = G_x[(V\setminus U) \cup
	\{s_1,s_2,s_3\}]+\{s_1s_2,s_1s_3,s_2s_3\}$ and, symmetrically, graph
	$G_2 = G_x[U\cup \{t_1,t_2,t_3\}]+\{t_1t_2,t_1t_3,t_2t_3\}$.  If both
	$G_1$ and $G_2$ have no bad 3-edge cuts, then we can apply Claim
	\ref{alpha1} to both $G_1$ and $G_2$.  For $G_1$, we find induced
	matchings $\{M^1_1, \ldots, M^1_5\}$ such that conditions (i) and (ii)
	hold.  Similarly, for $G_2$, we find induced matchings $\{M^2_1,
	\ldots, M^2_5\}$ such that (i) and (ii) hold.
	
	Notice that for each edge $e\in \{e_1,e_2,e_3\}$, there is exactly one
	induced matching in $\{M^1_1,\ldots,M^1_5\}$ and in
	$\{M^2_1,\ldots,M^2_5\}$ that contains $e_1$.  Without loss of
	generality, suppose $M^1_i$ and $M^2_i$ each contain edge $e_i$ for $i
	\in [3]$.  Then let $M_i = M^1_i\cup M^2_i$ for $i \in
	[5]$ and notice that $M_i$ is an induced matching in $G_x$.
	We conclude by induction on the number of bad cuts in $G_x$, since
	both $G_1$ and $G_2$ have fewer bad 3-edge cuts than does $G_x$.
\end{cproof}

\begin{claim}\label{alpha3}
	Suppose $G_x$ is 2-edge-connected. Then Lemma \ref{alpha}
	holds.
\end{claim}
\begin{cproof}
	We proceed by induction on the number of 2-edge cuts of
	$G_x$. If $G_x$ does not contain any 2-edge cuts then $G_x$ is
	3-edge-connected, so by Claim \ref{alpha2} the claim follows.
	
	For the induction step, consider 2-edge cut $\delta(U)=\{e_1,e_2\}$.
	Since $x$ is a half-cycle point, note that $e_1, e_2 \in W_x$.  Let
	$s_1$ and $s_2$ be the endpoints of $e_1$ and $e_2$ that are in $U$
	and let $t_1$ and $t_2$ be the other endpoints. 
	(Observe that neither $s_1s_2$ nor $t_1t_2$ is an edge in $G_x$;
	otherwise $G_x$ would contain a cut of $x$-value less than 2.)
	Consider graphs
	$G_1=G[U]+ s_1s_2$ and $G_2 = G[V\setminus U]+t_1t_2$.  The set of
	1-edges of $G_1$ is $\{W_x\cap E(G_1)\}\cup \{s_1s_2\}$, and the set of
	1-edges of $G_2$ is $\{W_x\cap E(G_2)\}\cup \{t_1t_2\}$.
	
	Without loss of generality, assume $r\in S$.  Apply induction on $G_1$
	to find induced matchings $\{M^1_1,\ldots,M^1_5\}$ where $s_1s_2\in
	M^1_1$, and on $G_2$ to obtain induced matchings
	$\{M^2_1,\ldots,M^2_5\}$ where $t_1t_2\in M^2_1$.  Set $M_1 =
	\{M^1_1\cup M^2_1 \cup \{e_1, e_2\} \}\setminus \{s_1s_2, t_1t_2\}$ and
	set $M_i = M^1_i\cup M^2_i$ for $i \in \{2, \ldots,5\}$.  Then
	$\{M_1,\ldots, M_5\}$ partition $W_x$ into induced matchings and
	satisfy conditions (i), (ii) and (iii).
\end{cproof}

The proof of Lemma \ref{alpha} follows from Claim \ref{alpha3}.

%% file: uniform-points.tex
\section{Construction of Tours for Uniform Points}\label{sec:uniform-points}

Recall the definition of $\frac{2}{k}$-uniform point from Section \ref{sec:uniformpts}. In this section we prove Theorem \ref{uc-improved} regarding $\frac{2}{3}$-uniform points and then we prove Theorem \ref{onlybase-intro} concerning $\frac{2}{4}$-uniform points.

  We would like to remark that a key idea in the proofs of both
  theorems is to first write the uniform point as a convex combination
  of points containing a perfect matching of 1-edges.  In the
  case of a $\frac{2}{3}$-uniform point, this results in a convex combination
  of $\frac{1}{2}$-cyclic points.  Using our main result (Theorem \ref{main}),
  we can save on 1-edges, which results on a uniform saving
  over all edges in the convex combination.  In the case
  of a $\frac{2}{4}$-uniform point, when we have an even number of
  vertices, we can obtain a convex combination of points where each
  vertex is incident to a 1-edge and three fractional $\frac{1}{3}$-edges.
  Then, if we have no tight 4-edge cuts, which occurs
  when the graph is essentially 6-edge-connected, we can save on
  1-edges. (Notice that the graph corresponding to a
  $\frac{2}{4}$-uniform point does not contain any 5-edge cuts.)
 This could serve as a base case if (i) the odd case is
  solved, and (ii) we could glue tours over tight 4-edge cuts.

\subsection{TSP on $\frac{2}{3}$-Uniform Points}\label{sec:thm1point5}

We start by the following lemma reducing TSP on $\frac{2}{3}$-uniform points  to TSP on $\frac{1}{2}$-cyclic points.
\begin{lemma}
	\label{lem:reduce}
	If for any $\frac{1}{2}$-cyclic point $x$ the vector $y$ defined as: $y_e = \frac{3}{2}-\epsilon$ for $e \in W_x$ and $y_e =
	\frac{3}{4} - \delta$ for $e \in H_x$ and $y_e = 0$ for $e\notin E_x$ for constants $\epsilon, \delta
	\geq 0$ belongs to $\tsp(K_n)$, then for any $\frac{2}{3}$-uniform point $z$ we have $ (\frac{3}{2} -
	\frac{\epsilon}{2} - \delta)z \in \tsp(K_n)$. 
\end{lemma}

\begin{proof}
	Let $z$ be a $\frac{2}{3}$-uniform point, and let $G_z=(V_n,E_z)$ be its support. Notice that $z \in \emptyset\join(G_z)$. Hence $z$ can be written as a convex combination of $\emptyset$-joins of $G_z$ denoted by $\{\lambda,\mathscr{C}\}$. Observe that each $\emptyset$-join $\mathcal{C}\in \mathscr{C}$ is in fact a 2-factor of $G_z$ since $z(\delta(u))=2$ and $|\mathcal{C}\cap \delta(u)|\leq 2$ for $u\in V_n$. For $\mathcal{C}\in \mathscr{C}$, we define $p^{\mathcal{C}}$ to be such that $p^{\mathcal{C}}_e = 1$ for $e\in \mathcal{C}$ and $p^{\mathcal{C}}_e = \frac{1}{2}$ for $e\in E_z\setminus \mathcal{C}$ and $p^{\mathcal{C}}_e=0$ for $e\in E_n\setminus E_z$. Notice that $p^{\mathcal{C}}$ is a $\frac{1}{2}$-cyclic point. Define $y^{\mathcal{C}}$ as follows: for $e\in E_n$ let $y^{\mathcal{C}}_e
	=\frac{3}{2}-\epsilon$ if $e\in W_{p^\mathcal{C}}$, and $y^{\mathcal{C}}_e =
	\frac{3}{4}-\delta$ if $e\in H_{p^{\mathcal{C}}}$ and $y_e =0$ otherwise. By assumption, we have
	$y^{\mathcal{C}}\in \tsp(K_n)$. Therefore,
	\begin{equation*}
	\hat{z}=\sum_{\mathcal{C}\in \mathscr{C}}\lambda_{\mathcal{C}} y^{\mathcal{C}} \in \tsp(K_n).
	\end{equation*}
	Observe that for $e\in E_x$
	\begin{align*}
	\hat{z}_e =&\frac{1}{3}\cdot (\frac{3}{2}-\epsilon)+ \frac{2}{3}\cdot (\frac{3}{4}-\delta)\\
	=& 1- \frac{\epsilon}{3}-\frac{2\delta}{3}\\
	=& (\frac{3}{2}-\frac{\epsilon}{2}-\delta)\cdot \frac{2}{3}=(\frac{3}{2}-\frac{\epsilon}{2}-\delta)\cdot x_e.
	\end{align*}
	Finally, for $e\in E_n\setminus E_x$, we have $\hat{z}_e = 0$. 
\end{proof}

A consequence of Theorem \ref{main} is that for $\frac{2}{3}$-uniform point $x\in \mathbb{R}^{E_n}$, we have $(\frac{3}{2}-\frac{1}{40})x\in \tsp(K_n)$. 

Haddadan, Newman and Ravi \cite{uniform} used the following theorem in \cite{bit13} to obtain the first factor below $\frac{3}{2}$ for approximating TSP on $\frac{2}{3}$-uniform point.
\begin{thm}[\cite{bit13}]\label{bit13}
	Let $G=(V,E)$ be a bridgeless cubic graph. Then $G$ has a 2-factor that covers all 3-edge cuts and 4-edge cuts of $G$.
	\end{thm}
We can combine the ideas in the proof of Theorem 1 of \cite{uniform} with Theorem \ref{main} to prove the following.
\begin{thm}	\label{application-uniformcover}
	Let $x$ be a $\frac{2}{3}$-uniform point and $G_x=(V_n,E_x)$ its support graph.  Then $\frac{17}{12}x\in \tsp(K_n)$.  If $G_x$ is Hamiltonian, then $\frac{87}{68} x\in \tsp(K_n)$. 
\end{thm}

\begin{proof}
	By Theorem \ref{bit13}, $G_x$ has a 2-factor $\mathcal{C}$ that
	covers all 3-edge cuts and 4-edge cuts of $G_x$. Define vector $y^1$ as
	follows: $y^1_e= 1$ for $e\in \mathcal{C}$ and $y^1_e = \frac{4}{5}$ for
	$e\in E_x\setminus \mathcal{C}$ and $y^1_e =0$ for $e\in E_n \setminus E_x$. 
	\begin{claim}
		We have $y^1\in \tsp(K_n)$.
	\end{claim}
	\begin{cproof}
		Notice graph $G'=G_x/\mathcal{C}$ is a 5-edge-connected graph. We can assume without loss of generality that $G'$ is also 5-regular\footnote{Replace every vertex $v$ of degree more than 5 with a doubled cycle of length $|\delta(v)|$ and connect each vertex in the cycle to a neighbor of $v$ in $G'$.}. For any vertex $r$ of $G'$ we have $\frac{2}{5}\chi^{E(G')}\in r\vtree(G')$. So the vector $\frac{2}{5}\chi^{E(G')}$ can be written as a convex combination of $r$-trees of $G'$ denoted by $\{\lambda,\mathcal{T}\}$. For $T\in \mathcal{T}$ the multigraph $F_T=\mathcal{C}+2T$ is a tour of $G_x$ and therefore $K_n$. Finally,  $y^1= \sum_{T\in \mathcal{T}}\lambda_T\chi^{F_T}$.
	\end{cproof}
	On the other hand, we can define $z$ where
	$z_e = \frac{1}{2}$ for $e\in \mathcal{C}$ and $z_e = 1$ for $e\in
	E_x\setminus \mathcal{C}$ for $z_e=0$ for $e\in E_n\setminus E_x$. Vector $z$ is a
	$\frac{1}{2}$-cyclic point, hence we can apply Theorem \ref{main} to
	obtain vector $y^2\in \tsp(K_n)$ such that $y^2_e = \frac{3}{4}$
	for $e\in \mathcal{C}$, $y^2_e = \frac{3}{2}-\frac{1}{20}$ for $e\in
	E_x\setminus \mathcal{C}$ and $y^2_e=0$ for $e\in E_n\setminus E_x$. Notice that
	$\frac{7}{9}y^1+ \frac{2}{9}y^2 \in \tsp(G)$ and is equal to $\frac{17}{12} x$.
	
	If $G_x$ is Hamiltonian, we can assume $\mathcal{C}$ is the Hamiltonian
	cycle of $G_x$.  Hence $\chi^{\mathcal{C}}\in \tsp(K_n)$. In this case $\frac{7}{17}\cdot \chi^{\mathcal{C}}+ \frac{10}{17}\cdot y^2 \in \tsp(K_n)$ and is equal to $\frac{87}{68} x$. 
\end{proof}

\subsection{A Base Case for $\frac{2}{4}$-Uniform Points}\label{sec:thm1point7}
Due to the fact that we can glue over critical cuts, we observed that TSP
on a $\theta$-cyclic point $x$ is essentially equivalent to the problem
with the assumption that $G_x$ contains no critical cuts.
Analogously, in the case of a $\frac{2}{4}$-uniform point $x$, Theorem
\ref{onlybase-intro} could serve as the base case if we were able
to glue over the proper minimum cuts of $G_x$.  However,
the difference here is that (1) the gluing arguments we presented for
$\theta$-cyclic points can not easily be extended to this case (due to the
increased complexity of the distribution of patterns), and (2) we
require an even number of vertices for our arguments.  

\introFourRegBasecase*

\begin{proof}
	We prove the claim by showing that there
	is a distribution of tours that satisfies the properties. It
	is easy to see that the proof yields a convex combination of
	tours of $G$.  Since $G$ does not have a proper 4-edge cut and since it is Eulerian, a proper cut of $G$ has at least 6 edges.
	
	Define $y_e=\frac{1}{4}$ for all $e\in E$. Vector $y$ is in
	the perfect matching polytope of $G$ and can be written as a
	convex combination of perfect matchings of $G$.  Choose a
	perfect matching $M$ at random from the distribution defined
	be the convex multipliers of this convex combination.
	
	Let $r\in V$. Define vector $z$ as follows: $z_e=1$ if $z\in M$ and
	$z_e = \frac{1}{3}$ for $z\in E\setminus M$. Observe that $z\in r\vtree(G_x)$ for any $v\in V$: $z(\delta(U))\geq \frac{1}{3}\cdot |\delta(U)|\geq 2$ if
	$|U|\geq 2$ and $|V\setminus U|\geq 2$ and $z(\delta(r))=2$.
	
	Applying Brook's theorem (similar to the proof of Lemma \ref{alpha})
	we can find collection $\{M_1,\ldots,M_7\}$ of induced matchings of
	$G$ that partition $M$. Choose $i \in[7]$ uniformly at random.
	For each $e=st\in M_i$, include the three edges incident on $s$ in one
	set and the three edges incident to $t$ in another set.  Notice all
	six edge are distinct since $G$ has no proper 4-edge cuts.  Apply
	Theorem \ref{boyd-sebo-rainbow} to decompose $z$ into a convex combination of
	rainbow $r$-trees of $G$ with respect to this partition. Take a
	random $r$-tree $T$ from this convex combination using the
	distribution defined by the convex multipliers. Let $O$ be the set of
	odd degree vertices of $T$.  Note that for each $e=st\in M_i$,
	$s,t\notin O$ by construction.  Define vector $p$ to be such that
	$p_e= \frac{1}{2}$ for $e\in M\setminus \{M_i\}$ and $p_e =
	\frac{1}{6}$ otherwise. We have $p\in O\join(G)$. Therefore, we can
	write $p$ as convex combination of $O$-joins of $G$. Choose one of
	the $O$-joins at random from the convex combination and label it
	$J$. Note that $F=T+J$ is a tour of $G$. For an edge $e\in M$ we
	have
	\begin{align*}
	\Pr[e\in J| e\in M] &= \Pr[e\in J|e\in M_i]\Pr[e\in M_i]+\Pr[e\in J|e\in M\setminus M_i]\Pr[e\in M\setminus M_i]\\
	&= \frac{1}{6}\cdot \frac{1}{7} + \frac{1}{2}\cdot \frac{6}{7} = \frac{19}{42}.
	\end{align*}
	If $e\notin M$, then we have $\Pr[e\in J|e\notin M]= \frac{1}{6}$. Hence,
	\begin{align*}
	\Pr[e\in J] &= \Pr[e\in J|e\in M]\Pr[e\in M]+\Pr[e\in J|e\notin M]\Pr[e\notin M]\\
	&= \frac{19}{42}\cdot
	\frac{1}{4} + \frac{1}{6}\cdot \frac{3}{4} = \frac{5}{21}.
	\end{align*}
	Observe that $\mathbb{E}[z_e] = 1\cdot \Pr[e\in M] + \frac{1}{3}\cdot Pr[e\notin M]=\frac{1}{2}$. Therefore, $\Pr[e\in T]=\Pr[e\notin T]=\frac{1}{2}$.
	\begin{align*}
	\mathbb{E}[\chi^F_e]&= 2\cdot  \Pr[e\in T \text{ and } e\in J] + \Pr[e\in T \text{ and } e\notin J]+ \Pr[e\notin T \text{ and } e\in J] \\
	&=2\cdot  \frac{1}{2} \cdot \frac{5}{21} + \frac{1}{2} \cdot \frac{16}{21}+ \frac{1}{2}\cdot \frac{5}{21}=\frac{3}{4} - \frac{1}{84}.
	\end{align*}
	Thus, each edge $e \in E$ is used
	to an extent $(\frac{3}{2}- \frac{1}{42})\cdot \frac{1}{2}$. This concludes the proof.
\end{proof}

%% file: conclusion.tex
\section{Concluding Remarks}\label{sec:conclusion}

In this paper, we showed how to improve the multiplicative
approximation factor of Christofides algorithm on the 1-edges of
$\theta$-cyclic points from $\frac{3}{2}$ to
$\frac{3}{2}-\frac{\theta}{10}$. Approaching Conjecture \ref{conj:4/3}
from this angle, we propose the following open problem, which is implied by the four-thirds conjecture. 

\begin{problem}\label{problem1}
Let $x\in \mathbb{R}^{E_n}$ be a $\theta$-cyclic point. Define vector $y$ as
follows: $y_e = \frac{4}{3}$ for $e \in
W_x$, $y_e = \frac{3}{2}x_e$ for $e\in H_x$ and $y_e = 0$ for
$e \notin E_x$.  Can we show $y\in \tsp(K_n)$?
\end{problem} 
In fact, the bound above is tight: for $\epsilon>0$, there exists a $\frac{1}{2}$-cyclic point $x^\epsilon$ such that vector $y^\epsilon$ defined as $y^{\epsilon}_e = \frac{4}{3}-\epsilon$ for $e \in
W_{x^\epsilon}$, $y_e = \frac{3}{2}x^{\epsilon}_e$ for $e\in H_{x^\epsilon}$ and $y_e = 0$ for
$e \notin E_{x^\epsilon}$.  Then $y^\epsilon\notin \tsp(K_n)$ (See Figure \ref{fig:4/3}). 
\input{envelope.tex}
 This
makes the problem above intriguing
even when restricted to $\frac{1}{2}$-cyclic points.  For a $\frac{1}{2}$-cyclic point $x$ where $H_x$ is disjoint union of 3-cycles, Boyd and Carr \cite{boydcarr} achieved this bound. Notice that this class includes $x^\epsilon$ in Figure \ref{fig:4/3}. Interestingly, a construction very similar to that of Boyd and Seb\H{o} \cite{boydsebo} implies that for a $\frac{1}{2}$-cyclic point
$x$ where $H_x$ is a union of vertex-disjoint 4-cycles, we can go beyond this factor.

\begin{thm}\label{boydsebo-mod}
	Let $x\in \mathbb{R}^{E_n}$ be a $\frac{1}{2}$-cyclic point where $H_x$ is a union of vertex-disjoint 4-cycles. Define vector $y$ as
	follows: $y_e = \frac{5}{4}$ for $e \in
	W_x$, $y_e = \frac{3}{2}x_e$ for $e\in H_x$ and $y_e = 0$ for
	$e \notin E_x$.  We have $y\in \tsp(K_n)$.
\end{thm}
\begin{proof}[Proof Sketch]
In this case, Boyd and Seb\H{o} \cite{boydsebo} showed that $G_x$ has a Hamilton cycle $H$ such that $H\subset W_x$ and $H$ intersects each 4-cycle of $H_x$ at opposite edges.
	
	They also show that $x$ can be written as convex combination
        of connectors of $G_x$ denoted by $\{\lambda,\mathcal{T}\}$
        such that for $T\in\mathcal{T}$ we have $|T\cap C|=2$ and $|T\cap (C\cap H)|=1$ for each 4-cycle $C\in H_x$~\cite{boydsebo}.
	
	For a $T\in \mathcal{T}$ define $O_T$ be the odd degree
        vertices of $T$ and define vector $z^T$ as follows: $z^T_e=
        \frac{1}{3}$ for $e\in W_x$, $z^T_e = \frac{1}{6}$ for $e\in
        H_x\cap H$ and $z^T_e = \frac{1}{2}$ for $e\in H_x\setminus
        H$. For $T\in \mathcal{T}$, we can show $z^T\in O_T\join(G_x)$ by
        following essentially the same arguments as Boyd and Seb\H{o}. This implies that $\chi^T+z^T\in \tsp(K_n)$. Therefore, $p=\sum_{T\in \mathcal{T}} \lambda_T(\chi^T+z^T)\in \tsp(K_n)$. In addition, we have $\chi^H\in \tsp(K_n)$. We conclude that $y=\frac{1}{4}\chi^H+ \frac{3}{4}p \in \tsp(K_n)$.
\end{proof}

Since the proof of Theorem \ref{boydsebo-mod} is essentially the same
as that in \cite{boydsebo}, it does not seem to extend to
$\theta$-cyclic points in which the fractional edges form $4$-cycles.


For $\frac{2}{4}$-uniform points it would be interesting to find a way
to apply the gluing approach, perhaps yielding further improvements to
the factors presented in \cite{karlin2019improved} and
\cite{gupta2021matroid}.  For $\frac{2}{k}$-uniform points with $k\geq
5$ nothing is known for TSP or 2EC beyond Christofides
$\frac{3}{2}$-approximation and the tiny step below $\frac{3}{2}$
implied by the recent work of Karlin et al.

\paragraph{Acknowledgements} We thank R. Ravi for his comments on the
presentation of this paper. The work of A. Haddadan is supported by
the U.S. Office of Naval Research under award number N00014-18-1-2099 and the U. S. National Science Foundation under award number CCF-1527032. The work of A. Newman is supported in part by IDEX-IRS SACRE.

%% file: envelope.tex
\begin{figure}[h]
	\centering
			\begin{tikzpicture}[scale=0.8]
		
		\draw [-] [black, line width=0.3mm] plot [smooth, tension=0] coordinates {(1,1) (2.4,1)};
		\draw [-] [black, line width=0.3mm] plot [smooth, tension=0] coordinates {(3.6,1) (5,1)};
		
		\draw [-] [black, line width=0.3mm] plot [smooth, tension=0] coordinates {(0,2) (2.4,2)};
		
		\draw [-] [black, line width=0.3mm] plot [smooth, tension=0] coordinates {(6,2) (3.6,2)};
			
		\draw [-] [black, line width=0.3mm] plot [smooth, tension=0] coordinates {(0,0) (2.4,0)};
		
		\draw [-] [black, line width=0.3mm] plot [smooth, tension=0] coordinates {(6,0) (3.6,0)};

		\draw [dashed] [black, line width=0.3mm,xshift=0 cm] plot [smooth, tension=0] coordinates {(0,0) (1,1)};
		
		\draw [dashed] [black, line width=0.3mm,xshift=0 cm] plot [smooth, tension=0] coordinates {(1,1) (0,2)};
		
		\draw [dashed] [black, line width=0.3mm,xshift=0 cm] plot [smooth, tension=0] coordinates {(0,0) (0,2)};
		
		\draw [dashed] [black, line width=0.3mm,xshift=0 cm] plot [smooth, tension=0] coordinates {(5,1) (6,2)};
		
		\draw [dashed] [black, line width=0.3mm] plot [smooth, tension=0] coordinates {(6,2) (6,0)};
		
		\draw [dashed] [black, line width=0.3mm] plot [smooth, tension=0] coordinates {(5,1) (6,0)};
		
		\draw[black,fill=white] (1,1) ellipse (0.07 cm  and 0.07 cm);
		\draw[black,fill=white] (0.6,2) ellipse (0.07 cm  and 0.07 cm);
		\draw[black,fill=white] (1.2,2) ellipse (0.07 cm  and 0.07 cm);
		\draw[black,fill=white] (5.4,2) ellipse (0.07 cm  and 0.07 cm);
		\draw[black,fill=white] (4.8,2) ellipse (0.07 cm  and 0.07 cm);
		\draw[black,fill=white] (4.2,2) ellipse (0.07 cm  and 0.07 cm);
		\draw[black,fill=white] (1.8,2) ellipse (0.07 cm  and 0.07 cm);
		\draw[black,fill=white] (3.6,2) ellipse (0.07 cm  and 0.07 cm);
		\draw[black,fill=white] (2.4,2) ellipse (0.07 cm  and 0.07 cm);
		
		\draw[black,fill=white] (1.7,1) ellipse (0.07 cm  and 0.07 cm);
		\draw[black,fill=white] (2.4,1) ellipse (0.07 cm  and 0.07 cm);
		\draw[black,fill=white] (4.3,1) ellipse (0.07 cm  and 0.07 cm);
		\draw[black,fill=white] (3.6,1) ellipse (0.07 cm  and 0.07 cm);
		
		\draw[black,fill=white] (0,2) ellipse (0.07 cm  and 0.07 cm);
		\draw[black,fill=white] (0,0) ellipse (0.07 cm  and 0.07 cm);
		\draw[black,fill=white] (5,1) ellipse (0.07 cm  and 0.07 cm);
		\draw[black,fill=white] (6,2) ellipse (0.07 cm  and 0.07 cm);
		\draw[black,fill=white] (6,0) ellipse (0.07 cm  and 0.07 cm);
		
		\draw[black,fill=white] (0.6,0) ellipse (0.07 cm  and 0.07 cm);
		\draw[black,fill=white] (1.2,0) ellipse (0.07 cm  and 0.07 cm);
		\draw[black,fill=white] (5.4,0) ellipse (0.07 cm  and 0.07 cm);
		\draw[black,fill=white] (4.8,0) ellipse (0.07 cm  and 0.07 cm);
		\draw[black,fill=white] (4.2,0) ellipse (0.07 cm  and 0.07 cm);
		\draw[black,fill=white] (1.8,0) ellipse (0.07 cm  and 0.07 cm);
		\draw[black,fill=white] (3.6,0) ellipse (0.07 cm  and 0.07 cm);
		\draw[black,fill=white] (2.4,0) ellipse (0.07 cm  and 0.07 cm);

		\node (Q) at (3,2) {{$\ldots$}};
		\node (Q) at (3,1) {{$\ldots$}};
		\node (Q) at (3,0) {{$\ldots$}};
		\end{tikzpicture}
	\label{fig:5.0}
	\caption{The support of $x^\epsilon$: In the figure above each of the three paths of solid edges contain $\ceil{\frac{1}{\epsilon}+1}$ vertices. We have $x^\epsilon_e= 1$ for solid edge $e$, $x^\epsilon_e = \frac{1}{2}$ for dashed edge $e$, and $x^\epsilon_e = 0$ for edge $e$ not depicted.} \label{fig:4/3}
\end{figure}
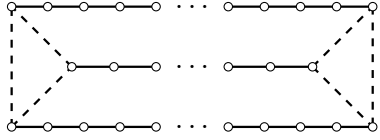